%% file: main.tex
\documentclass[sigplan,screen]{acmart}

\setcopyright{none}
\acmPrice{15.00}
\acmYear{2019}
\copyrightyear{2019}
\acmConference[PL]{PL}{2019}{USA}

\usepackage{booktabs}
\usepackage{subcaption}
\usepackage{algorithm}
\usepackage[noend]{algpseudocode}
\usepackage{amsmath}
\usepackage{amssymb}
\usepackage{stmaryrd}
\usepackage{url}
\usepackage{listings}
\usepackage{mdwlist}
\usepackage{pifont}
\usepackage{graphicx}
\usepackage{wrapfig}
\usepackage{fancyvrb}
\usepackage{paralist}
\usepackage{caption}
\usepackage{comment}
\usepackage{bbm, dsfont}
\usepackage{mathrsfs}
\usepackage{xcolor}
\usepackage{tikz}
\usepackage{pgfplots}
\usepackage{esvect}
\usepackage{xspace}
\usepackage{array, multirow}
\usepackage{rotating, makecell}
\usepackage{enumitem}
\usepackage{amssymb}

\input{macro}

\begin{document}

\title[Synthesizing Database Programs for Schema Refactoring]
{Synthesizing Database Programs \\ for Schema Refactoring}


\author{Yuepeng Wang}
\affiliation{
  \institution{University of Texas at Austin, USA}
}
\email{ypwang@cs.utexas.edu}
\author{James Dong}
\affiliation{
  \institution{University of Texas at Austin, USA}
}
\email{jdong@cs.utexas.edu}
\author{Rushi Shah}
\affiliation{
  \institution{University of Texas at Austin, USA}
}
\email{rshah@cs.utexas.edu}
\author{Isil Dillig}
\affiliation{
  \institution{University of Texas at Austin, USA}
}
\email{isil@cs.utexas.edu}

\input{abstract}



\maketitle

\input{intro}
\input{overview}

\input{prelim}

\input{synthesis}

\input{impl}
\input{eval}

\input{related}

\input{limit}
\input{concl}
\input{ack}

\bibliography{main}

\newcommand*{\extended}{}
\ifdefined\extended
\appendix
\input{appendixA}
\fi

\end{document}

%% file: macro.tex
\newcommand*{\comments}{}
\ifdefined\comments
\newcommand{\todo}[1]{\textcolor{red}{#1}}
\else
\newcommand{\todo}[1]{}
\fi

\newcommand{\irule}[2]%
   {\mkern-2mu\displaystyle\frac{#1}{\vphantom{,}#2}\mkern-2mu}
\newcommand{\irulelabel}[3]
{
\mkern-2mu
\begin{array}{ll}
\displaystyle\frac{#1}{\vphantom{,}#2} & \!\!\!\!#3
\end{array}
\mkern-2mu
}

\algnewcommand\algorithmicforeach{\textbf{for each}}
\algdef{S}[FOR]{ForEach}[1]{\algorithmicforeach\ #1\ \algorithmicdo}

\algnewcommand{\IfThenElse}[3]{
\State \algorithmicif\ #1\ \algorithmicthen\ #2\ \algorithmicelse\ #3}

\newcommand{\toolname}{\textsc{Migrator}\xspace}
\newcommand{\sketchtool}{\textsc{Sketch}\xspace}
\newcommand{\mediator}{\textsc{Mediator}\xspace}

\newcommand{\set}[1]{\{ #1 \}}

\newcommand{\proj}{\Pi}
\newcommand{\filter}{\sigma}
\newcommand{\equijoin}[2]{\vphantom{}_{#1\!\!\!}\Join_{#2}\!\!} 
\newcommand{\ins}{\texttt{ins}}
\newcommand{\del}{\texttt{del}}
\newcommand{\upd}{\texttt{upd}}

\newcommand{\SELECT}{~\textsf{\textbf{SELECT}}~}
\newcommand{\INSERT}{~\textsf{\textbf{INSERT}}~}
\newcommand{\DELETE}{~\textsf{\textbf{DELETE}}~}

\newcommand{\VALUES}{~\textsf{\textbf{VALUES}}~}
\newcommand{\WHERE}{~\textsf{\textbf{WHERE}}~}
\newcommand{\FROM}{~\textsf{\textbf{FROM}}~}
\newcommand{\JOIN}{~\textsf{\textbf{JOIN}}~}
\newcommand{\INTO}{~\textsf{\textbf{INTO}}~}
\newcommand{\ON}{~\textsf{\textbf{ON}}~}

\newcommand{\schema}{\mathcal{S}}
\newcommand{\prog}{\mathcal{P}}
\newcommand{\inst}{\mathcal{I}}
\newcommand{\sketch}{\Omega}
\newcommand{\valueCorr}{\Phi}

\newcommand{\pred}{\phi}
\newcommand{\hole}{\textbf{??}}

\newcommand{\view}{J}

\newcommand{\cstr}{\Psi} 
\newcommand{\model}{\mathcal{M}}
\newcommand{\cex}{\mathcal{E}}
\newcommand{\inputs}{\omega}
\newcommand*\circled[1]{\tikz[baseline=(char.base)]{
            \node[shape=circle,draw,line width=0.5pt,inner sep=1pt] (char) {#1};}}
\newcommand{\choice}{\,\vcenter{\hbox{\small \protect\circled{\tiny ?}}}\,}
\newcommand{\denot}[1]{\llbracket #1 \rrbracket}

%% file: abstract.tex
\begin{abstract}
Many  programs that interact with a database need to undergo \emph{schema refactoring} several times during their life cycle. Since this process  typically requires making significant changes to the program's implementation, schema refactoring is often non-trivial and error-prone. Motivated by this problem, we propose a new technique for \emph{automatically synthesizing} a new version of a database program given its original version and  the source and target schemas. Our method does not require  manual user guidance and ensures that the synthesized program is equivalent to the original one.   Furthermore, our method is quite efficient and can synthesize new versions of database programs (containing up to 263 functions) that are extracted from real-world web applications  with an average synthesis time of 69.4 seconds.
\end{abstract}

%% file: intro.tex
\section{Introduction}\label{sec:intro}

\emph{Database-driven applications} have been, and continue to be, enormously popular for web development. For example, most contemporary websites are built using database-driven applications in order to  generate webpage content dynamically. As a result, database applications  form the backbone of many industries, ranging from banking and e-commerce to telecommunications.

A common theme in the evolution of database applications is that they typically undergo \emph{schema refactoring} several times during their life cycle~\cite{refactordb,refactorsql}. Schema refactoring involves a change to the database schema, with the goal of improving the design and/or performance of the application \emph{without} changing its  semantics. Despite the frequent need to perform schema refactoring, this task is known to be non-trivial and error-prone~\cite{ambler2007test, wikitech}.
In particular, changes to the database schema often require re-implementing parts of the database program to make the program logic consistent with the underlying schema.
This task is especially non-trivial in the presence of \emph{structural} schema changes, such as those that involve splitting and merging relations or moving attributes between different tables.

While prior work has addressed the problem of \emph{verifying} equivalence between two database programs before and after schema refactoring~\cite{mediator}, generating a new version of the program after a schema change still remains an arduous and manual task. Motivated by this problem, this paper takes a step towards simplifying the evolution of programs that interact with a database. Specifically, we consider \emph{database programs} that consist of a set of  database transactions written in SQL. Given an existing database program $\prog$ that operates over source schema $\schema$ and a new target schema $\schema'$ that $\prog$ should be migrated to, our  method  automatically synthesizes a new database program $\prog'$ over the new schema $\schema'$ such that $\prog$ and $\prog'$ are \emph{semantically equivalent}. Thus, our  technique  automates the schema evolution process for these kinds of database programs while ensuring that no desirable behaviors are lost and no unwanted behaviors are introduced in the process.

Our methodology for automatically migrating database programs to a new schema is illustrated schematically in Figure~\ref{fig:overview}. Rather than synthesizing the new version of the program in one go, our algorithm decomposes the problem into three simpler sub-tasks, each of which leverages the results of the previous task in the pipeline. Specifically, given the source and the target schemas $\schema, \schema'$, our algorithm starts by guessing a candidate \emph{value correspondence} relating $\schema$ and $\schema'$. At a high level,  a value correspondence $\valueCorr$ specifies how attributes in $\schema'$ can be obtained using  the attributes in $\schema$~\cite{miller}.  Intuitively, learning a value correspondence is useful because (a) it is relatively easy to guess the correct correspondence based on attribute names in the schema, and (b) having a value correspondence dramatically constrains the space of programs that may be equivalent to the original program $\prog$.

\begin{figure}[t]
\vspace{5pt}
\begin{center}
\includegraphics[scale=0.35]{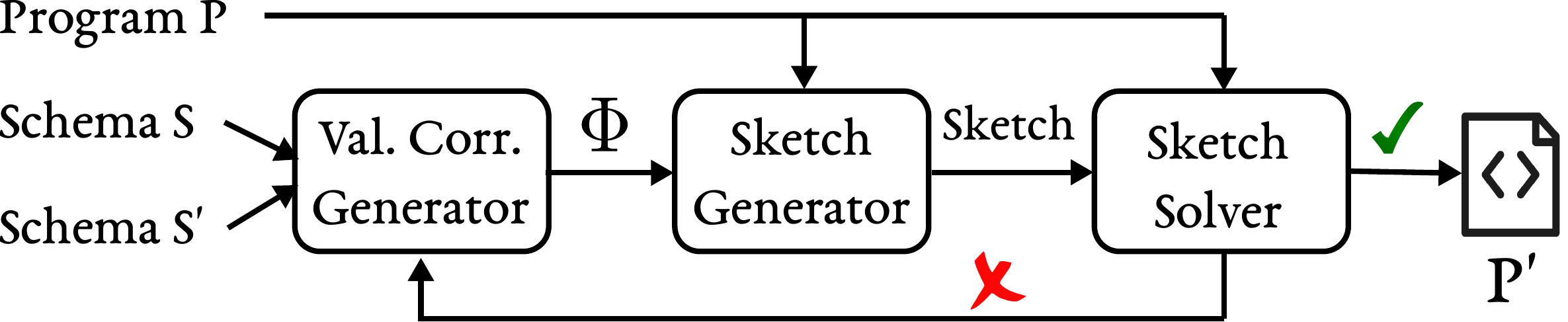}
\end{center}
\vspace{-5pt}
\caption{Synthesis methodology.}
\label{fig:overview}
\vspace{-10pt}
\end{figure}

While the value correspondence holds important clues as to what the transformation should look like, it nonetheless does not uniquely determine the target program $\prog'$. Thus, given a candidate value correspondence $\valueCorr$, our synthesis algorithm generates a \emph{program sketch} $\sketch$  that represents the space of all programs that \emph{may} be equivalent to the original program $\prog$ according to $\valueCorr$. In this context, a program sketch is a database program where some of the tables, attributes, or boolean constants are unknown. Furthermore, assuming  the  correctness of the candidate value correspondence $\valueCorr$, the sketch $\sketch$ is guaranteed to have a completion that is equivalent to $\prog$ (if one exists).

The third, and final, step in our synthesis pipeline ``solves'' the sketch $\sketch$ by finding an instantiation $\prog'$ of $\sketch$ that is equivalent to $\prog$. However, unlike existing sketch solvers that use the \emph{counterexample-guided inductive synthesis (CEGIS)} methodology, we use a different approach that does not require symbolically encoding the semantics of database programs into an SMT formula. Specifically, since  database query languages like SQL are not easily amenable to symbolic reasoning using established first-order theories supported by SMT solvers, our approach instead performs enumerative search over the space of all possible completions of the sketch. However, because this search space  is typically very large, a na\"{\i}ve search algorithm is difficult to scale to realistic database programs. Our approach deals with this difficulty by using a novel  algorithm that leverages \emph{minimum failing inputs (MFIs)} to dramatically prune the search space.

Overall, our synthesis algorithm for automatically migrating database programs to a new schema has several useful properties: First, it is completely push-button and does not require the user to provide anything other than the original program and the source and target schemas. Second, our approach is sound in that the synthesized program is provably equivalent to the original program and does not introduce any new, unwanted behaviors. Finally, since our method performs backtracking search over all possible value correspondences, it is guaranteed to find an equivalent program over the new schema if one exists.

We have implemented our proposed approach in a prototype tool called \toolname for automatically migrating database programs to a new schema. We evaluate \toolname on 20 benchmarks and show that it can successfully synthesize the new versions for \emph{all} twenty database programs with an average synthesis time of $69.4$ seconds per benchmark.  Thus, we believe these experiment results provide preliminary, but firm, evidence that the proposed synthesis technique can be useful to database program developers during the schema evolution process.

In all, this paper makes the following key contributions:

\begin{itemize}[leftmargin=*]
\item We propose a new synthesis technique for automatically migrating database programs to a new schema.
\item We describe a MaxSAT-based approach for lazily enumerating possible value correspondences between two schemas.
\item We describe a technique for generating program sketches from a given  value correspondence.
\item We propose a new  sketch solver based on symbolic search and conflict-driven learning from minimum failing inputs.
\item We evaluate the proposed technique on 20 schema refactoring scenarios and demonstrate that our  method can automate the desired migration task in all cases.
\end{itemize}

%% file: overview.tex
\section{Overview}

\begin{figure}[t]
\small
\[
\centering
\begin{array}{l}
\textbf{update}~ \emph{addInstructor(int id, String name, Binary pic)} \\
\ \ \ \ \INSERT \INTO \emph{Instructor} \VALUES \emph{(id, name, pic);} \\
\\[-5pt]
\textbf{update}~ \emph{deleteInstructor(int id)} \\
\ \ \ \ \DELETE \FROM \emph{Instructor} \WHERE \emph{InstId = id;} \\
\\[-5pt]
\textbf{query }~ \emph{getInstructorInfo(int id)} \\
\ \ \ \ \SELECT \emph{IName, IPic} \FROM \emph{Instructor} \WHERE \emph{InstId = id;} \\
\\[-5pt]
\textbf{update}~ \emph{addTA(int id, String name, Binary pic)} \\
\ \ \ \ \INSERT \INTO \emph{TA} \VALUES \emph{(id, name, pic);} \\
\\[-5pt]
\textbf{update}~ \emph{deleteTA(int id)} \\
\ \ \ \ \DELETE \FROM \emph{TA} \WHERE \emph{TaId = id;} \\
\\[-5pt]
\textbf{query }~ \emph{getTAInfo(int id)} \\
\ \ \ \ \SELECT \emph{TName, TPic} \FROM \emph{TA} \WHERE \emph{TaId = id;} \\
\end{array}
\]
\vspace{-10pt}
\caption{An example database program.}
\label{fig:example-prog}
\vspace{-10pt}
\end{figure}

In this section, we give an overview of our technique using a simple motivating example. Consider the database program shown in Figure~\ref{fig:example-prog} for managing and querying a course-related database with the following schema:
\[
\begin{array}{rl}
\emph{Class} & \!\!\!\!\! \emph{(ClassId, InstId, TaId)} \\
\emph{Instructor} & \!\!\!\!\! \emph{(InstId, IName, IPic)} \\
\emph{TA} & \!\!\!\!\! \emph{(TaId, TName, TPic)} \\
\end{array}
\]
This database has three tables that store information about courses, instructors, and TAs respectively. Here, the \emph{Instructor} and \emph{TA} tables store profile information about the course staff, including a picture. Since accessing a table containing large images may be potentially inefficient, the programmer decides to refactor the schema by introducing a new table for images. In particular, the desired new schema is as follows:
\[
\begin{array}{rl}
\emph{Class} & \!\!\!\!\! \emph{(ClassId, InstId, TaId)} \\
\emph{Instructor} & \!\!\!\!\! \emph{(InstId, IName, PicId)} \\
\emph{TA} & \!\!\!\!\! \emph{(TaId, TName, PicId)} \\
\emph{Picture} & \!\!\!\!\! \emph{(PicId, Pic)} \\
\end{array}
\]
As a result of this schema change, the program from Figure~\ref{fig:example-prog} needs to be re-implemented to conform to the new schema. We now explain how \toolname automatically synthesizes the new version of the program.

\paragraph{Value correspondence generation.} As mentioned in Section~\ref{sec:intro}, \toolname lazily enumerates possible value correspondences (VCs) between the source and target schemas. For this example, the first VC $\valueCorr$ generated by \toolname contains the following mappings:

\[
\begin{array}{rcl}
\emph{Instructor.IPic} & \rightarrow & \emph{Picture.Pic} \\
\emph{TA.TPic} & \rightarrow & \emph{Picture.Pic} \\
\end{array}
\]
In addition, all other attributes $T.a$ in the source schema are mapped to the same $T.a$ in the target schema.

\begin{figure}[t]
\small
\[
\centering
\begin{array}{l}
\textbf{update}~ \emph{addInstructor(int id, String name, Binary pic)} \\
\ \ \ \ \INSERT \INTO \emph{$\hole_1$ \{ Picture \!$\Join$\! Instructor, Picture \!$\Join$\! TA \!$\Join$\! Instructor,} \\
\ \ \ \ \ \ \ \ \emph{Picture $\Join$ TA $\Join$ Class $\Join$ Instructor \}} \VALUES \emph{(id, name, pic);} \\
\\[-3pt]
\textbf{update}~ \emph{deleteInstructor(int id)} \\
\ \ \ \ \DELETE \emph{$\hole_2$ \{ [Picture], $\ldots$, [Picture, Instructor, TA, Class] \} } \\
\ \ \ \ \ \ \ \ \FROM \emph{$\hole_3$ \{ Picture \!$\Join$\! Instructor, Picture \!$\Join$\! TA \!$\Join$\! Instructor,} \\
\ \ \ \ \ \ \ \ \emph{Picture $\Join$ TA $\Join$ Class $\Join$ Instructor \}} \WHERE \emph{InstId = id;} \\
\\[-3pt]
\textbf{query }~ \emph{getInstructorInfo(int id)} \\
\ \ \ \ \SELECT \emph{IName, Pic} \FROM \emph{$\hole_4$ \{} \\
\ \ \ \ \ \ \ \ \emph{Picture $\Join$ Instructor, Picture $\Join$ TA $\Join$ Instructor,} \\
\ \ \ \ \ \ \ \ \emph{Picture $\Join$ TA $\Join$ Class $\Join$ Instructor \}} \WHERE \emph{InstId = id;} \\
\\[-3pt]
\textbf{update}~ \emph{addTA(int id, String name, Binary pic)} \\
\ \ \ \ \INSERT \INTO \emph{$\hole_5$ \{ Picture $\Join$ TA, Picture $\Join$ Instructor $\Join$ TA,} \\
\ \ \ \ \ \ \ \ \emph{Picture $\Join$ Instructor $\Join$ Class $\Join$ TA \}} \VALUES \emph{(id, name, pic);} \\
\\[-3pt]
\textbf{update}~ \emph{deleteTA(int id)} \\
\ \ \ \ \DELETE \emph{$\hole_6$ \{ [Picture], $\ldots$, [Picture, Instructor, TA, Class] \} } \\
\ \ \ \ \ \ \ \ \FROM \emph{$\hole_7$ \{ Picture $\Join$ TA, Picture $\Join$ Instructor $\Join$ TA,} \\
\ \ \ \ \ \ \ \ \emph{Picture $\Join$ Instructor $\Join$ Class $\Join$ TA \}} \WHERE \emph{TaId = id;} \\
\\[-3pt]
\textbf{query }~ \emph{getTAInfo(int id)} \\
\ \ \ \ \SELECT \emph{TName, Pic} \FROM \emph{$\hole_8$ \{} \\
\ \ \ \ \ \ \ \ \emph{Picture $\Join$ TA, Picture $\Join$ Instructor $\Join$ TA,} \\
\ \ \ \ \ \ \ \ \emph{Picture $\Join$ Instructor $\Join$ Class $\Join$ TA \}} \WHERE \emph{TaId = id;} \\
\end{array}
\]
\vspace{-5pt}
\caption{Generated sketch over the new database schema.}
\label{fig:example-sketch}
\end{figure}

\paragraph{Sketch generation.} Next, \toolname uses the candidate VC $\valueCorr$ to generate a program sketch that encodes the space of all programs that are consistent with $\valueCorr$. The corresponding sketch for this example is shown in Figure~\ref{fig:example-sketch}. Here, each hole, denoted $\hole \set{c_1, \ldots, c_n}$, corresponds to an unknown constant drawn from the set $\set{c_1, \ldots, c_n}$. As will be discussed later in Section~\ref{sec:prelim}, we use the statement:
\[
\small
\INSERT \INTO ~ T_1 \Join T_2 ~ \VALUES \cdots
\]
as short-hand for:
\[
\small
\begin{array}{l}
\INSERT \INTO ~ T_1 ~ \VALUES \cdots \\
\INSERT \INTO ~ T_2 ~ \VALUES \cdots \\
\end{array}
\]
Thus, the first function  in the sketch corresponds to the following three possible implementations of \emph{addInstructor}:
\[
\small
\centering
\begin{array}{l}
\ \ \ \ \INSERT \INTO \emph{Instructor} \VALUES \emph{(id, name, $v_0$);} \\
\ \ \ \ \INSERT \INTO \emph{Picture} \VALUES \emph{($v_0$, pic);} \\
\ \ \ \ \emph{or} \\
\ \ \ \ \INSERT \INTO \emph{Instructor} \VALUES \emph{(id, name, $v_1$);} \\
\ \ \ \ \INSERT \INTO \emph{TA} \VALUES \emph{($v_2$, $v_3$, $v_1$);} \\
\ \ \ \ \INSERT \INTO \emph{Picture} \VALUES \emph{($v_1$, pic);} \\
\ \ \ \ \emph{or} \\
\ \ \ \ \INSERT \INTO \emph{Instructor} \VALUES \emph{(id, name, $v_4$);} \\
\ \ \ \ \INSERT \INTO \emph{Class} \VALUES \emph{($v_5$, id, $v_6$);} \\
\ \ \ \ \INSERT \INTO \emph{TA} \VALUES \emph{($v_6$, $v_7$, $v_4$);} \\
\ \ \ \ \INSERT \INTO \emph{Picture} \VALUES \emph{($v_4$, pic);} \\
\end{array}
\]
where $v_0, v_1, \ldots, v_7$ are unique values.

Observe that the program sketch shown in Figure~\ref{fig:example-sketch} has an enormous number of possible completions --- in particular, it corresponds to a search space of $164,025$ possible re-implementations of the original program.

\paragraph{Sketch completion.}  Given a sketch $\sketch$ and the original program $\prog$, the goal of sketch completion is to find an instantiation $\prog'$ of $\sketch$ such that $\prog'$ is equivalent to $\prog$, if such a $\prog'$ exists. Unfortunately, it is difficult to solve this sketch using existing solvers (e.g.,~\cite{sketch1,sketch2}) because the  symbolic encoding  of the program is quite complex due to the non-trivial semantics of SQL. In this paper, we deal with this difficulty by (a) encoding the space of \emph{all} possible programs represented by the sketch using a SAT formula $\cstr$, and (b) using minimum failing inputs to dramatically prune the search space represented by $\cstr$.

Going back to our sketch $\sketch$ from Figure~\ref{fig:example-sketch}, \toolname generates the following SAT formula that encodes all possible instantiations of $\sketch$:
\[
\begin{array}{c}
\oplus(b_1^1, b_1^2, b_1^3) \land \oplus(b_2^1, \ldots, b_2^{15}) \land
\oplus(b_3^1, b_3^2, b_3^3) \land \oplus(b_4^1, b_4^2, b_4^3) \land \\
\oplus(b_5^1, b_5^2, b_5^3) \land \oplus(b_6^1, \ldots, b_6^{15}) \land
\oplus(b_7^1, b_7^2, b_7^3) \land \oplus(b_8^1, b_8^2, b_8^3) \\
\end{array}
\]
Here, $\oplus$ denotes $n$-ary xor, and $b_i^j$ is a boolean variable that is assigned to true iff hole $\hole_i$ in the sketch is instantiated with the $j$-th constant in $\hole_i$'s domain.

Given this formula $\cstr$, \toolname queries the SAT solver for a model. For the purpose of this example, suppose the SAT solver returns the following model for $\cstr$:
\begin{equation}\label{eq:model1}
b_1^3 \land b_2^2 \land b_3^3 \land b_4^3 \land b_5^1 \land b_6^4 \land b_7^3 \land b_8^3
\end{equation}
which  corresponds to the following assignment  of the holes:
\begin{equation}
\small
\begin{array}{rl}
& \hole_1 = \hole_3 = \hole_4 = \emph{Picture} \Join \emph{TA} \Join \emph{Class} \Join \emph{Instructor} \\
\land & \hole_2 = [\emph{Instructor}] ~\land~ \hole_5 = \emph{Picture} \Join \emph{TA} ~\land~ \hole_6 = [\emph{TA}] \\
\land & \hole_7 = \hole_8 = \emph{Picture} \Join \emph{Instructor} \Join \emph{Class} \Join \emph{TA} \\
\end{array}
\label{eq:sat-formula}
\end{equation}

\begin{figure}[t]
\small
\[
\centering
\begin{array}{l}
\textbf{update}~ \emph{addInstructor(int id, String name, Binary pic)} \\
\ \ \ \ \INSERT \INTO \emph{Instructor} \VALUES \emph{(id, name, UID$_{\textit{0}}$);} \\
\ \ \ \ \INSERT \INTO \emph{Picture} \VALUES \emph{(UID$_{\textit{0}}$, pic);} \\
\\[-5pt]
\textbf{update}~ \emph{deleteInstructor(int id)} \\
\ \ \ \ \DELETE \emph{Instructor} \FROM \emph{Picture} \JOIN \emph{Instructor} \\
\ \ \ \ \ \ \ \ \ON \emph{Picture.PicId = Instructor.PicId} \WHERE \emph{InstId = id;} \\
\\[-5pt]
\textbf{query }~ \emph{getInstructorInfo(int id)} \\
\ \ \ \ \SELECT \emph{IName, Pic} \FROM \emph{Picture} \JOIN \emph{Instructor} \\
\ \ \ \ \ \ \ \ \ON \emph{Picture.PicId = Instructor.PicId} \WHERE \emph{InstId = id;} \\
\\[-5pt]
\textbf{update}~ \emph{addTA(int id, String name, Binary pic)} \\
\ \ \ \ \INSERT \INTO \emph{TA} \VALUES \emph{(id, name, UID$_{\textit{1}}$);} \\
\ \ \ \ \INSERT \INTO \emph{Picture} \VALUES \emph{(UID$_{\textit{1}}$, pic);} \\
\\[-5pt]
\textbf{update}~ \emph{deleteTA(int id)} \\
\ \ \ \ \DELETE \emph{TA} \FROM \emph{Picture} \JOIN \emph{TA} \\
\ \ \ \ \ON \emph{Picture.PicId = TA.PicId} \WHERE \emph{TaId = id;} \\
\\[-5pt]
\textbf{query }~ \emph{getTAInfo(int id)} \\
\ \ \ \ \SELECT \emph{TName, Pic} \FROM \emph{Picture} \JOIN \emph{TA} \\
\ \ \ \ \ \ \ \ \ON \emph{Picture.PicId = TA.PicId} \WHERE \emph{TaId = id;} \\
\end{array}
\]
\vspace{-5pt}
\caption{The synthesized database program.}
\label{fig:example-target}
\end{figure}

However, instantiating the sketch with this assignment results in a program $\prog'$ that is \emph{not} equivalent to $\prog$. Now, we \emph{could} block this program $\prog'$ by conjoining the  negation of Equation~\ref{eq:model1} with $\cstr$ and asking the SAT solver for another model. While this strategy would give us a different instantiation of sketch $\sketch$, it would preclude \emph{only one} of the $164,025$ possible instantiations of $\sketch$. Our key idea is to learn from this failure and block many other programs that are incorrect for the same reason as $\prog'$.

Towards this  goal, our approach  computes a \emph{minimum failing input}, which is a {shortest} sequence of function invocations such that the result of $\prog$ differs from that of $\prog'$. For this example, such a minimum failing input is the following invocation sequence $\inputs$:
\begin{equation}\label{eq:cex1}
 \emph{addTA}(\emph{ta1}, \emph{name1}, \emph{pic1}); \emph{getTAInfo}(\emph{ta1})
\end{equation}
This input establishes that $\prog'$ is \emph{not} equivalent to $\prog$ because the query result for $\prog$ is $\emph{(name1, pic1)}$ whereas the query result for $\prog'$ is empty.

Our idea is to utilize such a minimum failing input $\inputs$ to prune incorrect programs other than just $\prog'$. Specifically, let $\mathcal{F}$ denote the functions that appear in the invocation sequence $\inputs$, and let $\mathcal{H}$ be the holes that appear in the sketch for functions in $\mathcal{F}$. Our key intuition is that the assignments to holes in $\mathcal{H}$ are \emph{sufficient} for obtaining a spurious program, as $\inputs$ is a witness to the disequivalence between $\prog$ and $\prog'$.
 Thus, rather than blocking the whole model, we can extract the assignment to the holes in $\mathcal{H}$ and use this partial assignment to obtain a much stronger blocking clause. For our example, this yields the clause $\neg (b_5^1 \land b_8^3)$ because only the fifth and eighth holes appear in the sketches for \emph{addTA} and \emph{getTAInfo}. Using this blocking clause, we can eliminate a total of $18,225$ incorrect programs rather than just $\prog'$.

Continuing in this manner, \toolname finally obtains the following model for Equation~\ref{eq:sat-formula}:
\[
b_1^1 \land b_2^2 \land b_3^1 \land b_4^1 \land b_5^1 \land b_6^4 \land b_7^1 \land b_8^1
\]
This model corresponds to the program $\prog'$ shown in Figure~\ref{fig:example-target}, which is indeed equivalent to the original program from Figure~\ref{fig:example-prog}. Thus, \toolname returns $\prog'$ as the synthesis result.

%% file: prelim.tex
\section{Preliminaries}\label{sec:prelim}

In this section, we introduce the syntax and semantics of database programs and review what equivalence means in this context.

\subsection{Syntax and Semantics of Database Programs}

For the purpose of this paper, a database program consists of a set of functions, where each function is either a \emph{query} or \emph{update} to the database.  As shown in Figure~\ref{fig:syntax}, every function consists of a  name, a list of parameters, and a function body.

\input{fig-syntax}

The body of a query function is a relational algebra expression involving projection $(\proj)$, selection $(\filter)$, and join ($\Join$). As is standard, $\proj_{a_1, \ldots, a_n}(Q)$ recursively evaluates sub-query $Q$ to obtain a table $T$ and then constructs a table $T'$ that is the same as $T$ but containing only the columns $a_1, \ldots, a_n$. The filter operation $\filter_\phi(Q)$  recursively evaluates $Q$ to obtain a table $T$ and then filters out all rows in $T$ that do not satisfy predicate $\phi$.
A join expression $J_1 \equijoin{a_1}{a_2} J_2$ corresponds to the equi-join of $J_1$ and $J_2$ based on predicate $a_1 = a_2$, where $a_1$ is an attribute in $J_1$ and $a_2$ is an attribute in $J_2$. In the rest of this paper, we use the terminology \emph{join} or \emph{join chain} to refer to  both database tables as well as (possibly nested) join expressions of the form $J_1 \equijoin{a_1}{a_2} J_2$. Furthermore, since natural join is a special case of equi-join, we also use the standard notation $J_1 \Join J_2$ to denote natural joins where the equality check is implicit on identically named columns.

In contrast to query functions that do not change the state of the database, update functions can add or remove tuples to database tables. Specifically, an insert statement $\ins(T, \set{a_1:v_1, \ldots, a_n:v_n})$ inserts the tuple $\set{a_1:v_1, \ldots, a_n:v_n}$ into relation $T$.
To simplify presentation in the rest of the paper, we  use the syntax
\[
    \ins(T_1 \equijoin{\emph{fk}_1}{\emph{pk}_2} T_2, \set{a_1:v_1, \ldots, a_n:v_n, a'_1:v'_1, \ldots, a'_m : v'_m})
\]
as short-hand for the following sequence of insertions:
\[
\begin{array}{l}
\ins(T_1, ~ \set{\emph{pk}_1:u_0, ~ a_1: v_1, \ldots, a_n: v_n, ~ \emph{fk}_1: u_1}); \\
\ins(T_2, ~ \set{\emph{pk}_2:u_1, ~ a'_1: v'_1, \ldots, a'_m : v'_m}) \\
\end{array}
\]
where $u_0, u_1$ are unique values, and the schema for $T_1, T_2$ are $T_1(\emph{pk}_1, a_1, \ldots, a_n, \emph{fk})$ and $T_2(\emph{pk}_2, a'_1, \ldots, a'_m)$ respectively.

A delete statement $\del([T_1, \ldots, T_n], \view, \pred)$ removes from tables $T_1, \ldots, T_n$ exactly those tuples that satisfy predicate $\pred$ in join chain $\view$. As an example, consider the delete statement $\del([T_1], ~ T_1 \equijoin{a_1}{a_2} T_2, ~ \phi)$. Here, we first compute $T_1 \equijoin{a_1}{a_2} T_2$ to obtain a virtual table $T$ where each tuple in $T$ is the union of a source tuple in $T_1$ and a source tuple in $T_2$. We then obtain another virtual table $T'$ that filters out predicates satisfying $\phi$. Finally, we delete from $T_1$ all tuples that occur as (a prefix of) a tuple in $T'$. In contrast, if the statement is $\del([T_1, T_2], ~ T_1 \equijoin{a_1}{a_2} T_2, ~ \phi)$, the deletion is performed on both $T_1$ and $T_2$. We refer the reader to ~\cite{delete-join} for a more detailed discussion of the semantics of delete statements.~\footnote{We consider this form of delete statement rather than the more standard $\del(T, \phi)$ as it dramatically simplifies presentation in the rest of the paper.}

An update statement $\upd(\view, \pred, a, v)$ modifies the value of attribute $a$ to $v$ for all tuples satisfying predicate $\pred$ in join chain $\view$~\cite{update-join}. For instance, consider the update statement $\upd(T_1 \equijoin{a_1}{a_2} T_2, ~ \pred, ~ T_1.a_3, ~ v)$. Like delete statements, we first compute $T_1 \equijoin{a_1}{a_2} T_2$ and get a virtual table $T$ where each tuple in $T$ is the union of a source tuple in $T_1$ and a source tuple in $T_2$. Then we filter out tuples satisfying predicate $\pred$ in $T$ and get another virtual table $T'$. Finally, we update attribute $a_3$ in $T_1$  to value $v$ for all $T_1$ tuples that appear in $T'$.

\begin{example}
Consider a simple database with two tables:
\[
\vspace{-10pt}
\begin{tabular}{|c|c|c|}
\multicolumn{3}{c}{\textbf{Car}} \\
\hline
\textbf{cid} & \textbf{model} & \textbf{year} \\
\hline
1 & M1 & 2016 \\
\hline
2 & M2 & 2018 \\
\hline
\end{tabular}
~\ \ \ \ \ \ \ \ \ ~
\begin{tabular}{|c|c|c|}
\multicolumn{3}{c}{\textbf{Part}} \\
\hline
\textbf{name} & \textbf{amount} & \textbf{cid} \\
\hline
\text{tire} & 10 & 1 \\
\hline
\text{brake} & 20 & 1 \\
\hline
\text{tire} & 20 & 2 \\
\hline
\text{brake} & 30 & 2 \\
\hline
\end{tabular}
\]

The delete statement
\[
    \del([\text{Car, Part}], ~ \text{Car} \Join \text{Part}, ~ \text{model = M1})
\]
would delete tuple $(1, M1, 2016)$ from the Car table and tuples $(\text{tire}, 10, 1), (\text{brake}, 20, 1)$ from the Part table.
On the other hand, the update statement
\[
    \upd(\text{Car} \Join \text{Part}, ~ \text{model = M2} \land \text{name = tire}, ~ \text{amount}, ~ 30)
\]
would modify the third record of Part to $(\text{tire}, 30, 2)$.
\end{example}

\subsection{Equivalence of Two Database Programs}\label{sec:equiv}

Since our goal is to synthesize a program $\prog'$ that is equivalent to another database program $\prog$ with a different schema, we review the definition of equivalence introduced in prior work~\cite{mediator}.

Consider a database program $\prog$ over schema $\schema$ that has a set of update functions ${U} = (U_1, \ldots, U_n)$ and a set of query functions ${Q} = (Q_1, \ldots, Q_m)$. First, an \emph{invocation sequence} for $\prog$ is of the form
\[
\omega = (f_1, \sigma_1); \ldots; (f_{k-1}, \sigma_{k-1}); (f_k, \sigma_k)
\]
where $f_k$ is the name of a query function in $Q$, $f_1, \ldots, f_{k-1}$ refer to names of updates functions in $U$, and $\sigma_i$ corresponds to the arguments for function $f_i$.  Given a program $\prog$, we use the notation $\denot{\prog}_\omega$ to denote the result of executing $\prog$ on $\omega$.

Now, consider two programs $\prog, \prog'$ over schemas $\schema, \schema'$. Following ~\cite{mediator}, we say that $\prog$  is equivalent to $\prog'$, written $\prog \simeq \prog'$, if for \emph{any} invocation sequence $\omega$, we have $\denot{\prog}_\omega = \denot{\prog'}_\omega$ --- i.e., executing $\omega$ on $\prog$ yields the same query result as executing $\omega$ on $\prog'$ starting with an empty database instance. Thus, if two database programs are equivalent, then they yield the same query result after performing the same sequence of update operations on the database.

%% file: fig-syntax.tex
\begin{figure}[!t]
\[
\begin{array}{rcl}
    \emph{Prog} &:=& \emph{Func}+ \\
    \emph{Func} &:=& \textbf{update} ~ \emph{Name} (\emph{Param}+) ~ U \\
                &| & \textbf{query} ~ \emph{Name} (\emph{Param}+) ~ Q \\
    \emph{Update} ~ U &:=& \emph{InsStmt} ~|~ \emph{DelStmt} ~|~ \emph{UpdStmt} ~|~ U; U \\
    \emph{Query} ~ Q &:=& \proj_{a+}(Q) ~|~ \filter_\pred (Q) ~|~ J \\
    \emph{Join} ~ J &:=& T ~|~ J \equijoin{a}{a} J \\
    \emph{Pred} ~ \pred &:=& a \emph{ op } a ~|~ a \emph{ op } v ~|~ a \in Q ~|~ \pred \land \pred ~|~ \pred \lor \pred ~|~ \neg \pred \\
    \emph{InsStmt} &:=&  \ins(J, \set{(a:v) +}) \\
    \emph{DelStmt} &:=&  \del([T+], J, \pred) \\
    \emph{UpdStmt} &:=&   \upd(J, \pred, a, v) \\
\end{array}
\]
\[
\begin{array}{c}
\emph{Param} \in \textbf{Variable} \quad \emph{Name} \in \textbf{String} \\
T \in \textbf{Table} \quad a \in \textbf{Attribute} \quad v \in \textbf{Value} \cup \textbf{Variable} \\
\end{array}
\]
\vspace{-10pt}
\caption{Syntax of database programs. $+$ denotes the previous construct appears once or multiple times.}
\label{fig:syntax}
\end{figure}

%% file: synthesis.tex
\section{Synthesis Algorithm}\label{sec:synthesis}

In this section, we present our algorithm for automatically migrating database programs to a new schema. We start with an overview of the top-level algorithm and then discuss value correspondence enumeration, sketch generation, and sketch completion in more detail.

\subsection{Overview}

Our top-level synthesis algorithm is summarized as pseudo-code in Algorithm~\ref{algo:synthesis}. Given the original program $\prog$ over schema $\schema$ and the target schema $\schema'$, {\sc Synthesize} either returns a program $\prog'$ such that $\prog \simeq \prog'$ or $\bot$ to indicate that no equivalent program exists.

In a nutshell, the {\sc Synthesize} procedure is a while loop (lines 2 - 7) that lazily enumerates all possible \emph{value correspondences} between the source and target schemas. Formally, a value correspondence $\valueCorr$ from source schema $\schema$ to target schema $\schema'$ is a mapping from each attribute in $\schema$ to a \emph{set} of attributes in $\schema'$~\cite{miller}. Specifically, if $T'.b \in \valueCorr(T.a)$, this indicates that the entries in column $a$ in the source table $T$ are the same as the entries in column $b$ of table $T'$ in the target schema.
Observe that, if $\valueCorr$ maps some attribute $T.a$ in $\schema$ to $\emptyset$, this indicates that attribute $a$ of table $T$ has been deleted from the database. Similarly, if $|\valueCorr(T.a)| > 1$, this indicates that attribute $T.a$ has been duplicated in the target schema.
~\footnote{
Our notion of value correspondence is a slightly simplified  version of the definition given by Miller et al.~\cite{miller}. For example, their definition also allows  attributes in the target schema to be obtained by applying a function to attributes in the source schema.  Our technique can be extended to handle this scenario, albeit at the cost of increasing the size of the search space. 
}

\input{algo-synthesis}

Now, given a candidate value correspondence $\valueCorr$, the {\sc GenSketch} procedure at line 5 generates a sketch $\sketch$ that represents all programs that \emph{may} be equivalent to $\prog$ under the assumption that $\valueCorr$ is correct. Finally, the {\sc CompleteSketch} procedure (line 6) tries to find an instantiation $\prog'$ of $\sketch$ such that $\prog' \simeq \prog$. If such a $\prog'$ exists, then the algorithm terminates and returns $\prog'$ as the transformed program. On the other hand, if there is no completion  of the sketch that is equivalent to $\prog$, this indicates that the conjectured value correspondence is incorrect. In this case, the algorithm moves on to the next value correspondence $\valueCorr'$ and re-attempts the synthesis task using $\valueCorr'$. 

As formalized in more detail in Appendix A of this paper, our synthesis algorithm is both sound and relatively complete.
That is, if {\sc Synthesize} returns $\prog'$ as a solution, then $\prog'$ is indeed equivalent to $\prog$ by the definition from Section~\ref{sec:equiv}. Furthermore, {\sc Synthesize} is relatively complete, meaning that it can always find an equivalent program $\prog'$ under the assumption that (a) we have access to a sound and complete oracle for verifying equivalence of database programs, (b) $\prog'$ is related to $\prog$ according to a value correspondence that conforms to our definition, and (c) $\prog'$ has the same general structure as $\prog$.

In the following subsections, we explain the subroutines used in the {\sc Synthesize} algorithm in more detail.

\subsection{Lazy Enumeration of Value Correspondence}

In order to guarantee the completeness of our synthesis algorithm, we need a way to enumerate \emph{all} possible value correspondences between the source and target schemas. However, it is infeasible to generate all such value correspondences \emph{eagerly}, as there are exponentially many possibilities. In this section, we describe how to lazily enumerate value correspondences in decreasing order of likelihood using a partial weighted MaxSAT encoding.

\paragraph{Background on MaxSAT.} MaxSAT is a generalization of the traditional boolean satisfiability problem and aims to determine the maximum number of clauses that can be satisfied. Specifically, a MaxSAT problem is defined as a triple $(\mathcal{H}, \mathcal{S}, \mathcal{W})$, where  $\mathcal{H}$ is a set of \emph{hard clauses (constraints)}, $\mathcal{S}$ is a set of \emph{soft clauses}, and $\mathcal{W}$ is a mapping from each soft clause $c \in \mathcal{S}$ to a weight, which is an integer indicating the relative importance of satisfying clause $c$. Then, the goal of the MaxSAT problem is to find an interpretation $I$ such that:
\begin{enumerate}
\item $I$ satisfies all the hard clauses (i.e., $I \models \bigwedge_{c_i \in \mathcal{H}} c_i$)
\item $I$ maximizes the weight of the satisfied soft clauses
\end{enumerate}

\paragraph{Variables.} To describe our MaxSAT encoding, suppose that the source (resp. target) schema contains attributes $a_1, \ldots, a_n$ (resp. $a_1',\ldots, a_m'$). In our encoding, we introduce a boolean variable $x_{ij}$ to indicate that attribute $a_i$ in the source schema is mapped by the value correspondence $\valueCorr$ to attribute $a_j'$ in the target schema, i.e.,
\[
    x_{ij} \Leftrightarrow a'_j \in \valueCorr(a_i)
\]

\paragraph{Hard constraints.}
Hard constraints in our MaxSAT encoding rule out infeasible value correspondences:

\begin{itemize}[leftmargin=*]
\item \emph{Type-compatibility:} Since  $a_j' \in \valueCorr(a_i)$ indicates that the entries stored in $a_i$ and $a_j'$ are the same, $x_{ij}$ must be false if $a_i$ and $a_j'$ have different types. Thus, we add the following hard constraint for type compatibility:
\[
\bigwedge_{i,j} \neg x_{ij}  \ \emph{where} \  \emph{type}(a_i) \neq \emph{type}(a_j')
\]
\item \emph{Necessary condition for equivalence:} If the source program $\prog$ queries some attribute $a_i$ of the database, then there must be a corresponding attribute $a_j'$ that $a_i$ is mapped to; otherwise, the source and target programs would not be equivalent (recall Section~\ref{sec:equiv}). Thus, we introduce the following hard constraint:
\[
\bigvee_{1 \leq j \leq m} x_{ij} \ \emph{where} \  a_i \ \emph{is \ queried \  in} \ \prog
\]
which ensures that every attribute that is queried in the original program is mapped to at least one attribute in the target schema.
\end{itemize}

\paragraph{Soft constraints.}
The soft constraints in our encoding serve two purposes: First, since most attributes in the source schema typically have a unique corresponding attribute in the target schema, our soft constraints prioritize one-to-one mappings over one-to-many ones. Second, since attributes with similar names are more likely to be mapped to each other, they prioritize value correspondences that relate similarly named attributes.

To encode the latter constraint, we introduce a soft clause $x_{ij}$ with weight $\emph{sim}(a_i, a_j')$ for every variable $x_{ij}$. Here, \emph{sim} is a heuristic  metric that measures similarity between the names of attributes $a_i$ and $a_j'$.~\footnote{In our implementation, we implement \emph{sim} as $\alpha - \emph{Levenshtein}(a_i, a_j')$ where $\alpha$ is a fixed constant and $\emph{Levenshtein}$ is the standard Levenshtein distance.} To encode the former constraint,  we add a soft clause $ x_{ij} \rightarrow \neg x_{ik}$ (with fixed weight $\alpha$) for every $i \in [1, n], j \in [1, m]$ and $k \in (j, m]$. Essentially, such clauses tell the solver to de-prioritize mappings where the cardinality of $\valueCorr(a_i)$ is large.

\paragraph{Blocking clauses.}
While our initial MaxSAT encoding consists of exactly the hard and soft constraints discussed above, we need to add additional constraints to block previously rejected value correspondences. Specifically, let $A$ be an assignment (with corresponding value correspondence $\valueCorr_A$) returned by the MaxSAT solver, and suppose that there is no program $\prog'$ that is equivalent to $\prog$ under $\valueCorr_A$. In this case, our algorithm adds $\neg A$ as a hard constraint to prevent exploring the same value correspondence multiple times.

\subsection{Sketch Generation}\label{sec:sketch-gen}

In this section, we explain the {\sc GenSketch} procedure for generating a  sketch  that represents all programs that may be equivalent to $\prog$ under a given value correspondence $\valueCorr$. We first describe our sketch language and then explain how to use the value correspondence to generate a suitable sketch.

\input{fig-sketch-lang}

\paragraph{Sketch language.}
Our sketch language for database programs is presented in Figure~\ref{fig:sketch-lang} and differs from the source language in Figure~\ref{fig:syntax} in the following ways: First, programs in the sketch language can contain a construct of the form $\hole\set{e_1, \ldots, e_n}$, where the question mark is referred to as a \emph{hole} and the set of elements $\set{e_1, \ldots, e_n}$ is the \emph{domain} of that hole --- i.e., the question mark must be filled with some element drawn from the set $\{e_1, \ldots, e_n\}$. In addition, programs in the sketch language also contain a \emph{choice} construct $s_1 \choice s_2$, which is short-hand for the conditional statement:
\[
\textbf{if} \ \hole\set{\top, \bot} \ \textbf{then} \ s_1\   \textbf{else} \ s_2
\]
where $\top,\bot$ represent the boolean constants true and false, respectively. Thus, program sketches  in this context represent multiple (but finitely many) programs written in the syntax of Figure~\ref{fig:syntax}.

\paragraph{Join correspondence.} In order to generate a sketch from a  program $\prog$ and value correspondence $\valueCorr$, our approach first maps each join chain used in $\prog$ to a set of possible join chains over the target schema. We refer to such a mapping as a \emph{join correspondence} and say that a join correspondence $(\view, \view')$ is \emph{valid} with respect to  $\valueCorr$ if $\valueCorr$ can map all attributes used in $\view$ to attributes in $\view'$.

\input{fig-rules-view}

Figure~\ref{fig:rules-view} presents inference rules for checking whether a join correspondence $(\view, \view')$ is valid under $\valueCorr$. Specifically, the judgment $\valueCorr \vdash_A \view \sim \view'$ indicates that every attribute $a\in A$ of join chain $\view$ can be mapped to some attribute  of join chain $\view'$ under $\valueCorr$. Similarly, the judgment  $\valueCorr \vdash \view \sim \view'$ means that \emph{every} attribute in the join chain $\view$ can be mapped to an attribute of $\view'$ using $\valueCorr$.  Observe that, if $\valueCorr \vdash \view \sim \view_1$ \emph{and} $\valueCorr \vdash \view \sim \view_2$, this means that join chain $\view$ in the source program could map to \emph{either} $\view_1$ or $\view_2$ in the target program.

\paragraph{Sketching approach.} Our sketch generation technique uses the inferred join correspondences to produce a sketch that encodes all possible programs that may be equivalent to the source program. However, since a join chain $\view$ might correspond to any one of the join chains $\view_1, \ldots, \view_n$ in the target program, our sketch generation method proceeds in two phases: In the first phase, we non-deterministically pick any one of the join chains $\view_i$ that $\view$ could map to. Then, in the second phase, we combine the sketches obtained using $\view_1, \ldots, \view_n$ to obtain a more general sketch that accounts for every possibility.

\input{fig-rules-sketch}

\paragraph{Sketch generation, phase I.} The first phase of our sketch generation procedure is summarized in Figure~\ref{fig:rules-sketch} and assumes that every join chain $\view$ in the source program maps to a unique join chain $\view'$ in the target program.  Specifically, the rules in Figure~\ref{fig:rules-sketch} derive judgments of the form
$
\valueCorr \vdash s \leadsto \sketch
$,
meaning that statement $s$ in the original program can be rewritten into sketch $\sketch$ under the assumption that (a) $\valueCorr$ is correct and (b) every join chain in the source program corresponds to a unique join chain in the target program. We now explain each of these rules in more detail.

The  Attr (resp. Join) rule corresponds to a base case of  our inductive rewrite system and generates the sketch directly using the value (resp. join) correspondence. The Pred rule first generates  holes $h_1, \ldots, h_n$ for each attribute $a_i$ in $\phi$ and then generates a predicate sketch by replacing each $a_i$ with its corresponding sketch. The Filter and Proj rules are similar and generate the sketch by recursively rewriting the nested query, predicate, and attributes.

The last three rules in Figure~\ref{fig:rules-sketch} generate sketches for update statements. Here, the Update and Insert rules are  straightforward and generate the sketch by recursively rewriting the nested attributes and predicates. For the Delete rule, recall that deletion statements are of the form $\texttt{del}(\emph{Tbls}, \view, \phi)$, where \emph{Tbls} can refer to any non-empty subset of the tables used in $\view$. Thus, the sketch for deletion statements contains a hole for \emph{Tbls}, with the domain of the hole being the power-set of the tables used in $\view'$.

\input{fig-rules-compose}

\paragraph{Sketch generation, phase II.} Recall that a join chain in the source program may correspond to multiple join chains in the target schema --- i.e., the target join chain is not \emph{uniquely} determined by a given value correspondence. Thus, the second phase of our  algorithm combines the sketches  generated during the first phase to synthesize a more general sketch that accounts for this ambiguity.

\input{fig-op-comp}

Figure~\ref{fig:rules-compose} describes the second phase of sketch generation using judgments of the form $\valueCorr \vdash s \twoheadrightarrow \sketch$. At a high level, the rules in Figure~\ref{fig:rules-compose} compose the sketches obtained during the first phase to obtain a more general sketch. To start with, the Lift rule corresponds to a base case and states that the $\twoheadrightarrow$ relation is initially obtained using the $\leadsto$ relation. The Query rule composes multiple sketches $\sketch_1, \ldots, \sketch_n$ for a query statement $Q$ as $\sketch_1 \choice \ldots \choice \sketch_n$ --- i.e., the composed sketch is a union of the individual sketches.

The Update rule is similar to Query, but it is slightly more involved. In particular, suppose that we have two different sketches $\sketch_1, \sketch_2$ for an update statement $U$. Now, we need to account for the possibility that either one or \emph{both} of the updates may happen. Thus,  the corresponding sketch for update statements is $\sketch_1 \ \choice \ \sketch_2 \ \choice \ (\sketch_1; \sketch_2)$ rather than the simpler sketch $\sketch_1 \choice \sketch_2$ for query statements. The Update rule generalizes this discussion to arbitrarily many sketches by using a binary operator  $\bullet$ (defined in Figure~\ref{fig:op-comp}) that distributes sequential composition ($;$) over the choice ($\choice$) construct. Finally, the Seq rule allows generating a sketch for $U_1; U_2$ using the sketch $\sketch_i$ for each $U_i$.

Given a statement $s$ in the source program, its corresponding sketch $\sketch$ is obtained by applying the rewrite rules from Figure~\ref{fig:rules-compose} to a fixed-point. Specifically, let $\sketch_1, \ldots, \sketch_n$ be the set of sketches such that $\valueCorr \vdash s \twoheadrightarrow \sketch_1, \ \ldots, \ \valueCorr \vdash s \twoheadrightarrow \sketch_n$, and let us say that a sketch $\sketch$ is more general than $\sketch'$, written $\sketch \succeq \sketch'$, if $\sketch$ represents more programs than $\sketch'$.  Then, the resulting sketch for $s$ is the most general sketch $\sketch_i $ such that $\forall j \in [1,n]. \sketch_i \succeq \sketch_j$.

\subsection{Sketch Completion}\label{sec:sketch-complete}

In this section, we explain our algorithm for solving the database program sketches from Section~\ref{sec:sketch-gen}. As mentioned earlier, we do not encode the precise semantics of the sketch using an SMT formula because  relational algebra operators are difficult to express using standard first-order theories supported by SMT solvers. Instead, we perform symbolic search (using SAT) over the space of programs encoded by the sketch and then subsequently check equivalence. If the two programs are not equivalent, we employ \emph{minimum failing inputs} to further prune the search space by identifying programs that share the same root cause of failure as a previously encountered program.

\input{algo-completion}

\paragraph{Overview.} Our sketch completion procedure is  summarized in Algorithm~\ref{algo:completion} and takes as input a program sketch $\sketch$ together with the source program $\prog$. The output of {\sc CompleteSketch} is either a completion  $\prog'$ of $\sketch$ such that $\prog \simeq \prog'$ or $\bot$ to indicate no such program exists.

At a high level, the \textsc{CompleteSketch} procedure first generates a boolean formula $\cstr$ that represents \emph{all} possible completions of the sketch $\sketch$ (line 2). While any model of $\cstr$ corresponds to a concrete program $\prog'$ that is an instantiation of $\sketch$, such a program $\prog'$ may or may not be equivalent to the input program $\prog$. Thus, the sketch solving algorithm enters a loop (lines 3--9) that lazily explores different instantiations of $\sketch$, checks equivalence, and adds useful blocking  clauses to the SAT encoding $\cstr$ as needed. In what follows, we explain the algorithm (and its subroutines) in more detail.

\paragraph{Initial SAT encoding.} The goal of the {\sc Encode} procedure at line 2 is to generate a SAT formula that encodes all possible completions of $\sketch$. Specifically, for each hole $\hole_i\set{e_1, \ldots, e_n}$ in the sketch, we introduce $n$ boolean variables $b_i^1, \ldots, b_i^n$ such that $b_i^j = \emph{true}$  if and only if hole $\hole_i$ is instantiated with expression $e_j$.~\footnote{Since the choice construct $s_1 \choice s_2$ is just syntactic sugar for $\textbf{if} \ \hole\set{\top, \bot} \ \textbf{then} \ s_1\   \textbf{else} \ s_2$, we assume it has been de-sugared before this SAT encoding.} Since any valid completion of sketch $\sketch$ must assign every hole $
\hole_i$ to some expression $e_j$ in its domain, our initial SAT encoding is obtained as follows:

\[
    \Psi = \bigwedge_{\hole_i \in \emph{Holes}(\sketch)} \oplus(b_i^1, \ldots, b_i^{i_n})
\]
where the domain of $\hole_i$ consists of expressions $e_1, \ldots e_{i_n}$, and $\oplus$ denotes the $n$-ary xor operator. Observe that every model $\model$  of formula $\Psi$ corresponds to one particular instantiation of $\sketch$; thus, the procedure \textsf{Instantiate}  produces program $\prog'$ by assigning hole $\hole_i$ to expression $e_j$ if and only if $\model$ assigns $b_i^j$ to true.

\paragraph{Verification and blocking clauses.} As is apparent from the discussion above, our symbolic encoding $\Psi$  of the sketch intentionally does not enforce equivalence between the source and target programs. Thus, whenever we obtain a completion $\prog'$ of the sketch, we must check whether $\prog, \prog'$ are actually equivalent using the \textsf{Verify} subroutine  at line 6 of Algorithm~\ref{algo:completion}. If the two programs are indeed equivalent, the algorithm terminates with $\prog'$ as a solution. Otherwise, in the next iteration, we   ask the SAT solver for a different model, which corresponds to a different instantiation of the input sketch. However, in practice, there are an enormous number (e.g., up to $10^{39}$) of completions of the sketch; thus, a synthesis algorithm that tests equivalence for every possible sketch completion is unlikely to scale. Our sketch completion algorithm addresses this issue by using minimum failing inputs to block \emph{many} programs at the same time.

Specifically, a \emph{minimum failing input} for a pair of programs $\prog, \prog'$ is an invocation sequence $\omega$ (recall Section~\ref{sec:equiv}) satisfying the following criteria:

\begin{enumerate}
\item We have $\denot{\prog}_{\omega} \neq \denot{\prog'}_{\omega}$. That is, $\omega$ is a witness to the disequivalence of $\prog$ and $\prog'$
\item  There does not exist another invocation sequence $\omega'$ such that $|\omega'| < |\omega|$ and $\denot{\prog}_{\omega'} \neq \denot{\prog'}_{\omega'}$
\end{enumerate}

Intuitively, minimum failing inputs are useful in this context because they provide feedback about which assignments to which holes cause program $\prog'$ to \emph{not} be equivalent to $\prog$. Specifically, let $\mathcal{H}$ (resp. $\overline{\mathcal{H}}$) be the holes used in functions that appear (resp. do \emph{not} appear) in $\omega$, and let $A_\mathcal{H}$ denote the assignments to holes $\mathcal{H}$. Then, any program that instantiates $\sketch$ by assigning $A_\mathcal{H}$ to $\mathcal{H}$ will also be incorrect, regardless of the assignments to holes $\overline{\mathcal{H}}$. Our sketch completion algorithm uses this observation to rule out  many programs beyond $\prog'$. Specifically, let $\mathcal{H} = \{\hole_1, \ldots, \hole_n\}$ and suppose that $A_\mathcal{H}$ assigns expression $e_{k_i}$ to each $\hole_i$. Then, the \textsf{Block} procedure (line 9 of Algorithm~\ref{algo:completion}) generates the following blocking clause:
\[
 \varphi =  \neg (b_1^{k_1} \land \ldots \land b_n^{k_n})
\]
Intuitively, this blocking clause $\varphi$ rules out all completions of $\sketch$ that agree with $\prog'$ on the assignment to holes in $\mathcal{H}$. Since minimum failing inputs typically involve a small subset of the methods in the program, this technique allows us to rule out \emph{many} programs in one iteration. Furthermore, as we discuss in Section~\ref{sec:impl}, minimum failing inputs are inexpensive to obtain using testing.

%% file: algo-synthesis.tex
\begin{figure}[!t]
\vspace{-10pt}
\begin{algorithm}[H]
\caption{Synthesizing database programs}
\label{algo:synthesis}
\begin{algorithmic}[1]
\Procedure{\textsc{Synthesize}}{$\prog, \schema, \schema'$}
\vspace{2pt}
\small
\Statex \textbf{Input:} Program $\prog$ over source schema $\schema$, target schema $\schema'$
\Statex \textbf{Output:} Program $\prog'$ or $\bot$ to indicate failure
\vspace{2pt}

\While{\emph{true}}
    \State $\valueCorr \leftarrow \textsc{NextValueCorr}(\schema, \schema')$;
    \If{$\valueCorr = \bot$} \Return $\bot$; \EndIf
    \State $\sketch \leftarrow \textsc{GenSketch}(\valueCorr, \prog)$;
    \State $\prog' \leftarrow \textsc{CompleteSketch}(\sketch, \prog)$;
    \If{$\prog' \neq \bot$} \Return $\prog'$; \EndIf
\EndWhile

\EndProcedure
\end{algorithmic}
\end{algorithm}
\vspace{-10pt}
\end{figure}

%% file: fig-sketch-lang.tex
\begin{figure}[!t]
\[
\begin{array}{rcl}
   \emph{Prog} &:=& \emph{Func}+ \\
    \emph{Func} &:=& \textbf{update} ~ \emph{Name} (\emph{Param}+) ~ U \\
                &| & \textbf{query} ~ \emph{Name} (\emph{Param}+) ~ Q \\
    \emph{Update} ~ U &:=& \emph{InsStmt} ~|~ \emph{DelStmt} ~|~ \emph{UpdStmt} ~|~ U; U ~|~ U \choice U \\
    \emph{Query} ~ Q &:=& \proj_{(\hole\set{a+})+}(Q) ~|~ \filter_\pred (Q) ~|~ \view ~|~ Q \choice Q \\
    \emph{Join} ~ \view &:=&  T ~|~ \view \equijoin{a}{a} \view \\
    \emph{Pred} ~ \pred &:=& \hole\set{a+} \emph{ op } \hole\set{a+} ~|~ \hole\set{a+} \emph{ op } v \\
                        & |& \pred \land \pred ~|~ \pred \lor \pred ~|~ \neg \pred \\
    \emph{InsStmt} &:=& \ins(\view, ~ \set{(\hole\set{a+}:v) +}) \\
    \emph{DelStmt} &:=& \del(\hole\set{L+}, ~ \view, ~ \pred) \\
    \emph{UpdStmt} &:=& \upd(\view, ~ \pred, ~ \hole\set{a+}, ~ v) \\
    \emph{TabList} ~ L &:=& [T+] \\
\end{array}
\]
\[
\begin{array}{c}
\emph{Param} \in \textbf{Variable} \quad \emph{Name} \in \textbf{String} \\
T \in \textbf{Table} \quad a \in \textbf{Attribute} \quad v \in \textbf{Value} \cup \textbf{Variable} \\
\end{array}
\]
\caption{Sketch Language. $\hole$ represents a hole in the sketch and the subsequent set indicates the domain of that hole. $\choice$ is a choice operator and $s_1 \choice s_2$ denotes the statement could either be $s_1$ or $s_2$. $E+$ indicates a list of elements of type $E$.}
\label{fig:sketch-lang}
\end{figure}

%% file: fig-rules-view.tex
\begin{figure}[!t]
\[
\begin{array}{c}
\irulelabel
{\begin{array}{c}
    A \subseteq \emph{Attrs}(\view) \quad \forall a \in A.~ \exists a' \in \valueCorr(a). ~ a' \in \emph{Attrs}(\view') \\
\end{array}}
{\valueCorr \vdash_A \view \sim \view'}
{\textrm{(Attrs)}} \\ \ \\

\irulelabel
{A = \emph{Attrs}(\view) \quad \valueCorr \vdash_A \view \sim \view'}
{\valueCorr \vdash \view \sim \view'}
{\textrm{(JoinChain)}}

\end{array}
\]
\vspace{-10pt}
\caption{Inference rules for checking join correspondence $(\view, \view')$  under value correspondence $\valueCorr$.}
\label{fig:rules-view}
\vspace{-8pt}
\end{figure}

%% file: fig-rules-sketch.tex
\begin{figure}[!t]
\[
\!\!\!\!\begin{array}{c}
\irulelabel
{\valueCorr \vdash \view \sim \view'}
{\valueCorr \vdash \view \leadsto \view'}
{\textrm{(Join)}}
~
\irulelabel
{\valueCorr(a) = \set{a'_1, \ldots, a'_n}}
{\valueCorr \vdash a \leadsto \hole\set{a'_1, \ldots, a'_n}}
{\textrm{(Attr)}} \\ \ \\

\irulelabel
{a_i \in \emph{Attrs}(\pred) \quad \valueCorr \vdash a_i \leadsto h_i \quad i = 1, \ldots, n}
{\valueCorr \vdash \pred \leadsto \pred[h_1/a_1, \ldots, h_n/a_n]}
{\textrm{(Pred)}} \\ \ \\

\irulelabel
{\valueCorr \vdash Q \leadsto \sketch \quad \valueCorr \vdash \pred \leadsto \pred'}
{\valueCorr \vdash \filter_{\pred}(Q) \leadsto \filter_{\pred'}(\sketch)}
{\textrm{(Filter)}} \\ \ \\

\irulelabel
{\begin{array}{c}
    \valueCorr \vdash Q(\view) \leadsto \sketch(h) \quad \valueCorr \vdash a_j \leadsto h_j \quad j = 1, \ldots, m \\
    A = \set{a_1, \ldots, a_m} \!\cup\! \emph{Attrs}(Q) \ \ \ \valueCorr \vdash_A \view \sim \view' \\
\end{array}}
{\valueCorr \vdash \proj_{a_1, \ldots, a_m}(Q(\view)) \leadsto \proj_{h_1, \ldots, h_m}(\sketch(\view'))}
{\textrm{(Proj)}} \\ \ \\

\irulelabel
{\begin{array}{c}
    A = \emph{Attrs}(L) \cup \emph{Attrs}(\pred) \quad \valueCorr \vdash \pred \leadsto \pred' \\
    \quad \valueCorr \vdash_A \view \sim \view' \quad \emph{TabLists}(\view') = \set{L_1, \ldots, L_n} \\
\end{array}}
{\valueCorr \vdash \del(L, \view, \pred) \leadsto \del(\hole \set{L_1, \ldots, L_n}, \view', \pred')}
{\textrm{(Delete)}} \\ \ \\

\irulelabel
{\begin{array}{c}
    \valueCorr \vdash \pred \leadsto \pred' \quad \valueCorr \vdash a \leadsto h \\
    A = \emph{Attrs}(\pred) \cup \set{a} \quad \valueCorr \vdash_A \view \sim \view' \\
\end{array}}
{\valueCorr \vdash \upd(\view, \pred, a, v) \leadsto \upd(\view', \pred', h, v)}
{\textrm{(Update)}} \\ \ \\

\irulelabel
{\begin{array}{c}
    \valueCorr \vdash \view \sim \view' \quad \valueCorr \vdash a_i \leadsto h_i \quad i = 1, \ldots, n \\
\end{array}}
{\begin{array}{rl}
    \valueCorr \vdash & \ins(\view, \set{a_1 : v_1, \ldots, a_m : v_m}) \leadsto \\
    & \ins(\view', \set{h_1 : v_1, \ldots, h_m : v_m}) \\
\end{array}}
{\textrm{(Insert)}}

\end{array}
\]
\vspace{-10pt}
\caption{Rewrite rules for generating sketch from value correspondence $\valueCorr$. All holes $\hole$ are annotated with an index to ensure they are globally unique. The function $\emph{TabLists}$ returns all non-empty subset of tables in a join, i.e. $\emph{TabLists}(T_1 \Join \ldots \Join T_n) = \emph{PowerSet}(\set{T_1, \ldots, T_n}) \setminus \emptyset$.}
\label{fig:rules-sketch}
\end{figure}

%% file: fig-rules-compose.tex
\begin{figure}[!t]
\[
\!\!\!\!\begin{array}{c}

\irulelabel
{\valueCorr \vdash s \leadsto \sketch}
{\valueCorr \vdash s \twoheadrightarrow \sketch}
{\textrm{(Lift)}}
 \\ \ \\
\irulelabel
{\begin{array}{c}
    \valueCorr \vdash Q \twoheadrightarrow \sketch \quad \valueCorr \vdash Q \leadsto \sketch' \\
    \sketch = \sketch_1 \choice \ldots \choice \sketch_n \quad \sketch' \neq \sketch_i \quad i = 1, \ldots, n \\
\end{array}}
{\valueCorr \vdash Q \twoheadrightarrow \sketch \choice \sketch'}
{\textrm{(Query)}} \\ \ \\
\irulelabel
{\begin{array}{c}
    \valueCorr \vdash U \twoheadrightarrow \sketch \quad \valueCorr \vdash U \leadsto \sketch' \\
    \sketch = \sketch_1 \choice \ldots \choice \sketch_n \quad \sketch' \neq \sketch_i \quad i = 1, \ldots, n \\
\end{array}}
{\valueCorr \vdash U \twoheadrightarrow \sketch \choice \sketch' \choice (\sketch \bullet \sketch')}
{\textrm{(Update)}} \\ \ \\

\irulelabel
{\valueCorr \vdash U_1 \twoheadrightarrow \sketch_1 \quad \valueCorr \vdash U_2 \twoheadrightarrow \sketch_2}
{\valueCorr \vdash U_1; U_2 \twoheadrightarrow \sketch_1; \sketch_2}
{\textrm{(Seq)}}

\end{array}
\]
\vspace{-10pt}
\caption{Inference rules for composing multiple sketches. The composition operator $\bullet$ is defined in Figure~\ref{fig:op-comp}.}
\label{fig:rules-compose}
\end{figure}

%% file: fig-op-comp.tex
\begin{figure}[!t]
\[
\begin{array}{rcl}
U_1 \bullet U_2 &=& U_1; U_2 \quad (U_1 = \ins \emph{ or }~ \del \emph{ or }~ \upd) \\
(U_1 ; U_2) \bullet U_3 &=& U_1; U_2; U_3 \\
(U_1 \choice U_2) \bullet U_3 &=& (U_1 \bullet U_3) \choice (U_2 \bullet U_3) \\
\end{array}
\]
\vspace{-10pt}
\caption{Definition of the composition operator.}
\label{fig:op-comp}
\end{figure}

%% file: algo-completion.tex
\begin{figure}[!t]
\vspace{-0.2in}
\begin{algorithm}[H]
\caption{Sketch Completion}
\label{algo:completion}
\begin{algorithmic}[1]
\Procedure{\textsc{CompleteSketch}}{$\sketch, \prog$}
\vspace{2pt}
\small
\Statex \textbf{Input:} Sketch $\sketch$, Source program $\prog$
\Statex \textbf{Output:} Target program $\prog'$ or $\bot$ to indicate failure
\vspace{2pt}

\State $\cstr \leftarrow \textsc{Encode}(\sketch)$;
\While{\textsf{SAT}($\cstr$)}
    \State $\model \leftarrow \textsf{GetModel}(\cstr)$;
    \State $\prog' \leftarrow \textsf{Instantiate}(\sketch, \model)$;
    \State $\emph{done} \leftarrow \textsf{Verify}(\prog, \prog')$;
    \If{\emph{done}} \Return $\prog'$; \EndIf
    \State $\cex \leftarrow \textsf{MinCex}(\prog, \prog')$;
    \State $\cstr \leftarrow \cstr \land \textsf{Block}(\model, \cex)$;
\EndWhile
\State \Return $\bot$;

\EndProcedure
\end{algorithmic}
\end{algorithm}
\vspace{-0.3in}
\end{figure}

%% file: impl.tex
\section{Implementation}\label{sec:impl}

We have implemented the proposed synthesis technique  in a new tool called \toolname, which is implemented in Java. \toolname uses the Sat4J solver~\cite{sat4j} for answering all SAT and MaxSAT queries and the \mediator tool~\cite{mediator} for verifying equivalence between a pair of database programs.
In the remainder of this section, we  discuss two important design choices about our implementation.

\paragraph{Sketch generation.} Recall from Section~\ref{sec:sketch-gen} that our sketch generation algorithm produces a sketch using a so-called \emph{join correspondence},  which in turn is synthesized from a candidate value correspondence. While our presentation in Section~\ref{sec:sketch-gen} presents ``type-checking'' rules that determine whether a join correspondence is valid with respect to some value correspondence, it can be inefficient to consider all possible join chains in the target schema and then check whether they are feasible. Thus, rather than taking an enumerate-then-check approach,
our implementation \emph{algorithmically} produces join correspondences that are feasible with respect to a given value correspondence.

To see how we infer all target join chains that may correspond to a source join chain $J$, suppose we are given a  value correspondence $\valueCorr$ and let $A$ be the  set of attributes that occur in $\view$. Our goal is to find all join chains $J_1, \ldots, J_n$ over the target schema such that for every attribute $a \in A$, there is a corresponding attribute $a' \in \emph{Attrs}(J_i)$. We reduce the problem of finding all such possible join chains  to the problem of finding all possible Steiner trees~\cite{steinertree} over a graph data structure where nodes represent tables and edges represent join-ability relations.

In more detail, let $A'$ be a set of attributes over the target schema such that for every $a \in A$, there exists some $a' \in A'$ where $a' \in \valueCorr(a)$, and let $\mathcal{T'}$ denote the set of tables containing all attributes in $A'$. Since the source join chain refers to all attributes in $A$, we need to find exactly those join chains over the target schema that ``cover'' the relations in which $A'$ appears. Towards this goal, we construct a graph data structure $G = (V, E)$ as follows: The nodes $V$ are tables in the target schema, and there is an edge $(T, T')$ if tables $T$ and $T'$ can be joined with each other. Now, recall that, given a graph $G = (V, E)$ and a set of vertices $V' \subseteq V$, a Steiner tree is a connected subgraph that spans all vertices $V'$. Since our goal is to ``cover'' exactly the tables $\mathcal{T'}$ in the target schema, we compute all possible Steiner trees spanning $\mathcal{T'}$ and convert them to join chains  in the expected way.

\paragraph{Generating minimum failing inputs.} Recall from Section~\ref{sec:sketch-complete} that our sketch completion algorithm uses minimum failing inputs to prune the search space. In our implementation, we generate such inputs using a bounded testing procedure.
Specifically, we generate a fixed set of constants for each type (e.g., $\set{0, 1}$ for integers) as the seed set to be used for arguments. Then, given such a seed set $C$ of constants, our testing engine generates all possible invocation sequences containing only constants from $C$ in increasing order of length.
For each  invocation sequence $\omega$, we execute both $\prog$ and $\prog'$ on $\omega$ and check if the outputs are different. If so, we return $\omega$ as a minimum failing input, and otherwise, we test equivalence using  the next invocation sequence.

\paragraph{Verification.} Our sketch completion algorithm from Section~\ref{sec:sketch-complete} invokes a \textsf{Verify} procedure to check if two programs are equivalent.  However, since full-fledged verification using the \mediator tool~\cite{mediator} can be quite expensive, we first perform exhaustive testing up to some bound and invoke \mediator only when no failing inputs are found.
In principle, it is possible that the testing procedure fails to find a failing input while the verifier cannot establish  equivalence.  We have not encountered this kind of  scenario in practice, but it could nonetheless happen in theory.

%% file: eval.tex
\section{Evaluation}

\input{tab-results}

To evaluate the proposed idea, we use \toolname to  automatically migrate 20 database programs to a new schema.

\paragraph{Benchmarks.}
All 20 programs in our benchmark set
are taken from prior work~\cite{mediator} for verifying equivalence between database programs.~\footnote{
While  prior work considers 21 benchmarks in the evaluation, one of these benchmarks cannot be verified by \mediator. Since we use \mediator as our verifier, we exclude this one benchmark from our evaluation.
}
Specifically, half of these benchmarks are adapted from textbooks and online tutorials, and the remaining half are manually extracted from real-world web applications on Github.
However, because the input language of \toolname is slightly different from that of \mediator, we write a translator to convert the database programs to \toolname's input language.

\paragraph{Experimental Setup.}
All of our experiments are conducted on a machine with Intel Xeon(R) E5-1620 v3 quad-core CPU and 32GB of physical memory, running the Ubuntu 14.04 operating system. For each synthesis benchmark, we set a time limit of 24 hours.

\subsection{Main Results}

Our main experimental results are summarized in Table~\ref{tab:results}. Here,  the first ten rows correspond to benchmarks taken from  database schema refactoring textbooks, and the latter  ten rows correspond to real-world Ruby-on-Rails applications collected from Github. The ``Description'' column in Table~\ref{tab:results} explains how the database schema differs between the source and target versions, and ``Funcs'' shows the number of functions that need to be synthesized.
The next two columns under ``Source Schema'' (resp. ``Target Schema'') describe the number of tables and attributes in the source (resp. target) schema.
The last four columns report the results obtained by running \toolname on each benchmark. Specifically, the column  ``Value Corr'' shows the number of value correspondences considered by \toolname, and ``Iters'' shows the number of  programs explored before an equivalent one is found. Finally, the ``Synth Time'' column shows  synthesis time in seconds (excluding verification), and ``Total Time'' shows total time, including both synthesis and verification.

The key takeaway message from this experiment is that \toolname can successfully synthesize equivalent versions of all 20 benchmarks, including the database programs in real-world Ruby-on-Rails web applications with up to $263$ functions. Furthermore,  synthesis time (excluding verification) ranges from $0.3$ seconds to $1304.7$ seconds, with the average time being $69.4$ seconds in total or $1.2$ seconds per function. We believe these results provide strong evidence that our proposed technique can be quite useful for automating the schema refactoring process for database programs.

\subsection{Comparison to Baselines}

Given that there are other existing techniques for solving program sketches, we also evaluate our sketch completion algorithm by comparing our method against two baselines. In particular, our first baseline is the \sketchtool tool~\cite{sketch1}, and the second one is a variant of our own sketch completion algorithm that does not use minimum failing inputs (MFIs).

\paragraph{Comparison with \sketchtool.} To compare our approach against the \sketchtool tool~\cite{sketch1}, we first implemented the semantics of SQL in \sketchtool by encoding each SQL statement as a C function. Specifically, our \sketchtool encoding models each database table as an array of arrays, with the nested array representing a tuple, and we model each SQL operation as a function that reads and updates the array as appropriate.

The results of this experiment are summarized in Table~\ref{tab:comp-sketch}. The main observation is that {\sc Sketch} times out on all real-world benchmarks  from Github as well as two textbook examples, namely Oracle-2 and Ambler-8. For all other benchmarks, \toolname is significantly faster than {\sc Sketch}, with speed-ups ranging between $5.3$x to $10455.0$x in terms of synthesis time.
~\footnote{
Since \sketchtool only performs bounded model checking rather than  full-fledged verification, we only report  speedup in terms of synthesis time rather than total time including verification. The speedup in terms of total time (including verification) ranges from $2.4$x to $1358.0$x.
}
We believe this experiment demonstrates the advantage of our proposed sketch completion algorithm compared to the standard CEGIS approach  implemented in~\sketchtool.

\input{tab-comp-sketch}

\paragraph{Comparison with enumerative search.} Since the key novelty of our sketch completion algorithm is the use of minimum failing inputs to prune the search space, we also compare our approach against a baseline that does not use MFIs. In particular, this baseline uses the same SAT encoding of the search space but blocks only a single program  at a time. More concretely, given a model $\model$ of the SAT encoding $\cstr$, this baseline updates $\cstr$ by conjoining $\neg \model$ whenever verification fails. Effectively, this baseline performs enumerative search but does so in a symbolic way using a SAT solver.

The results of this experiment are summarized in Table~\ref{tab:comp-enum}. As we can see from this table, the impact of MFIs is particularly pronounced for the Ambler-8 textbook example and almost all real-world benchmarks. In particular, \toolname is $192.3$x faster than enumerative search on average. Moreover, without using MFIs to prune the search space, two of the benchmarks do not terminate within a time-limit of 24 hours. Hence, these results demonstrate that our MFI-based sketch completion  is very important for practical synthesis.

\input{tab-comp-enum}

%% file: tab-results.tex
\begin{table*}[t]
\centering
\vspace{5pt}
\caption{Main experimental results.}
\label{tab:results}
\vspace{-5pt}
\small
\begin{tabular}{|c|l|l|c|c|c|c|c|c|c|c|c|}
\hline
& \multirow{2}{*}{\textbf{Benchmark}} & \multirow{2}{*}{\textbf{Description}} & \multirow{2}{*}{\!\textbf{Funcs}\!} & \multicolumn{2}{c|}{\textbf{Source Schema}} & \multicolumn{2}{c|}{\textbf{Target Schema}} & \textbf{Value} & \multirow{2}{*}{\!\textbf{Iters}\!} & \textbf{Synth} & \textbf{Total} \\
\cline{5-8}
& & & & \textbf{Tables} & \textbf{Attrs} & \textbf{Tables} & \textbf{Attrs} & \textbf{Corr} & & \textbf{Time(s)} & \textbf{Time(s)} \\
\hline
\hline
\multirow{11}{*}{\begin{turn} {90}\makecell{\small \bf textbook \ bench}\end{turn}}
& Oracle-1 & Merge tables & 4 & 2 & 8 & 1 & 6 & 1 & 1 & 0.3 & 2.7 \\
\cline{2-12}
& Oracle-2 & Split tables & 19 & 3 & 17 & 7 & 25 & 1 & 5 & 0.5 & 11.3 \\
\cline{2-12}
& Ambler-1 & Split tables & 10 & 1 & 6 & 2 & 7 & 1 & 2 & 0.3 & 2.9 \\
\cline{2-12}
& Ambler-2 & Merge tables & 10 & 2 & 7 & 1 & 6 & 1 & 1 & 0.3 & 0.6 \\
\cline{2-12}
& Ambler-3 & Move attrs & 7 & 2 & 5 & 2 & 5 & 2 & 5 & 0.4 & 30.6 \\
\cline{2-12}
& Ambler-4 & Rename attrs & 5 & 1 & 2 & 1 & 2 & 1 & 1 & 0.3 & 0.5 \\
\cline{2-12}
& Ambler-5 & Add associative tables & 8 & 2 & 5 & 3 & 6 & 5 & 7 & 0.3 & 3.1 \\
\cline{2-12}
& Ambler-6 & Replace keys & 10 & 2 & 9 & 2 & 8 & 1 & 1 & 0.3 & 0.7 \\
\cline{2-12}
& Ambler-7 & Add attrs & 8 & 2 & 7 & 2 & 8 & 1 & 1 & 0.3 & 0.6 \\
\cline{2-12}
& Ambler-8 & Denormalization & 14 & 3 & 10 & 3 & 13 & 1 & 7 & 0.5 & 3.1 \\
\hline
\hline
\multirow{10}{*}{\begin{turn} {90}\makecell{\small \bf real-world \ bench}\end{turn}}
& cdx & Rename attrs, split tables & 138 & 16 & 125 & 17 & 131 & 1 & 7 & 11.9 & 38.9 \\
\cline{2-12}
& coachup & Split tables & 45 & 4 & 51 & 5 & 55 & 1 & 10 & 1.8 & 6.7 \\
\cline{2-12}
& 2030Club & Split tables & 125 & 15 & 155 & 16 & 159 & 1 & 2 & 5.2 & 24.8 \\
\cline{2-12}
& rails-ecomm & Split tables, add new attrs & 65 & 8 & 69 & 9 & 75 & 1 & 6 & 2.5 & 10.3 \\
\cline{2-12}
& royk & Add and move attrs & 151 & 19 & 152 & 19 & 155 & 1 & 17 & 46.1 & 60.1 \\
\cline{2-12}
& MathHotSpot & Rename tables, move attrs & 54 & 7 & 38 & 8 & 42 & 6 & 11 & 1.2 & 5.8 \\
\cline{2-12}
& gallery & Split tables & 58 & 7 & 52 & 8 & 57 & 1 & 11 & 2.5 & 9.4 \\
\cline{2-12}
& DeeJBase & Rename attrs, split tables & 70 & 10 & 92 & 11 & 97 & 1 & 8 & 3.5 & 9.3 \\
\cline{2-12}
& visible-closet & Split tables & 263 & 26 & 248 & 27 & 252 & 1 & 108 & 1304.7 & 1370.8 \\
\cline{2-12}
& probable-engine\!\! & Merge tables & 85 & 12 & 83 & 11 & 78 & 1 & 9 & 4.6 & 17.5 \\
\hline
\hline
& \textbf{Average} & - & \textbf{57.5} & \textbf{7.2} & \textbf{57.1} & \textbf{7.8} & \textbf{59.4} & \textbf{1.5} & \textbf{11.0} & \textbf{69.4} & \textbf{80.5} \\
\hline
\end{tabular}
\vspace{5pt}
\end{table*}

%% file: tab-comp-sketch.tex
\begin{table}[t]
\centering
\vspace{5pt}
\caption{Comparison with \sketchtool.}
\label{tab:comp-sketch}
\vspace{-5pt}
\small
\begin{tabular}{|c|l|c|c|}
\hline
& \multirow{2}{*}{\textbf{Benchmark}} & \multicolumn{2}{c|}{\textbf{\sketchtool}} \\
\cline{3-4}
& & \textbf{Synth Time(s)} & \textbf{Speedup} \\
\hline
\hline
\multirow{11}{*}{\begin{turn} {90}\makecell{\small \bf textbook \ bench}\end{turn}}
& Oracle-1 & 88.2 & 294.0x \\
\cline{2-4}
& Oracle-2 & >86400.0 & >172800.0x \\
\cline{2-4}
& Ambler-1 & 3136.5 & 10455.0x \\
\cline{2-4}
& Ambler-2 & 71.5 & 238.3x \\
\cline{2-4}
& Ambler-3 & 74.7 & 186.8.5x \\
\cline{2-4}
& Ambler-4 & 1.6 & 5.3x\\
\cline{2-4}
& Ambler-5 & 494.4 & 1648.0x \\
\cline{2-4}
& Ambler-6 & 226.2 & 754.0x \\
\cline{2-4}
& Ambler-7 & 814.8 & 2716.0x \\
\cline{2-4}
& Ambler-8 & >86400.0 & >172800.0x \\
\hline
\hline
\multirow{10}{*}{\begin{turn} {90}\makecell{\small \bf real-world \ bench}\end{turn}}
& cdx & >86400.0 & >7260.5x \\
\cline{2-4}
& coachup & >86400.0 & >48000.0x \\
\cline{2-4}
& 2030Club & >86400.0 & >16615.4x \\
\cline{2-4}
& rails-ecomm & >86400.0 & >34560.0x \\
\cline{2-4}
& royk & >86400.0 & >1874.2x \\
\cline{2-4}
& MathHotSpot & >86400.0 & >72000.0x \\
\cline{2-4}
& gallery & >86400.0 & >34560.0x \\
\cline{2-4}
& DeeJBase & >86400.0 & >24685.7x \\
\cline{2-4}
& visible-closet & >86400.0 & >66.2x \\
\cline{2-4}
& probable-engine & >86400.0 & >18782.6x \\
\hline
\hline
& \textbf{Average} & \textbf{>52085.4} & \textbf{>750.5x} \\
\hline
\end{tabular}
\vspace{5pt}
\end{table}

%% file: tab-comp-enum.tex
\begin{table}[t]
\centering
\vspace{5pt}
\caption{Comparison with symbolic enumerative search.}
\label{tab:comp-enum}
\vspace{-5pt}
\small
\begin{tabular}{|c|l|c|c|c|}
\hline
& \multirow{2}{*}{\textbf{Benchmark}} & \multicolumn{3}{c|}{\textbf{Symbolic Enum}} \\
\cline{3-5}
& & \textbf{Iters} & \!\textbf{Synth Time(s)}\! & \!\textbf{Speedup}\! \\
\hline
\hline
\multirow{11}{*}{\begin{turn} {90}\makecell{\small \bf textbook \ bench}\end{turn}}
& Oracle-1 & 1 & 0.3 & 1.0x \\
\cline{2-5}
& Oracle-2 & 5 & 0.5 & 1.0x \\
\cline{2-5}
& Ambler-1 & 2 & 0.3 & 1.0x \\
\cline{2-5}
& Ambler-2 & 1 & 0.3 & 1.0x \\
\cline{2-5}
& Ambler-3 & 6 & 0.4 & 1.0x \\
\cline{2-5}
& Ambler-4 & 1 & 0.3 & 1.0x \\
\cline{2-5}
& Ambler-5 & 11 & 0.4 & 1.3x \\
\cline{2-5}
& Ambler-6 & 1 & 0.3 & 1.0x \\
\cline{2-5}
& Ambler-7 & 1 & 0.3 & 1.0x \\
\cline{2-5}
& Ambler-8 & 67996 & 54367.6 & 108735.2x \\
\hline
\hline
\multirow{10}{*}{\begin{turn} {90}\makecell{\small \bf real-world \ bench}\end{turn}}
& cdx & 5595 & 6169.4 & 518.4x \\
\cline{2-5}
& coachup & 1303 & 76.2 & 42.3x \\
\cline{2-5}
& 2030Club & 2 & 5.2 & 1.0x \\
\cline{2-5}
& rails-ecomm & 2779 & 602.5 & 241.0x \\
\cline{2-5}
& royk & >31249 & >86400.0 & >1874.2x \\
\cline{2-5}
& MathHotSpot & 115 & 5.3 & 4.4x \\
\cline{2-5}
& gallery & 21483 & 32266.2 & 12906.5x \\
\cline{2-5}
& DeeJBase & 605 & 142.8 & 40.8x \\
\cline{2-5}
& visible-closet & >9512 & >86400.0 & >66.2x \\
\cline{2-5}
& probable-engine\!\! & 1661 & 540.3 & 117.5x \\
\hline
\hline
& \textbf{Average} & \!\!\textbf{>7116.5}\!\! & \textbf{>13348.9} & \textbf{>192.3x} \\
\hline
\end{tabular}
\vspace{5pt}
\end{table}

%% file: related.tex
\section{Related Work}

In this section, we survey related papers on program synthesis, schema refactoring, and analysis of database applications.

\paragraph{Schema evolution.} There is a body of literature  on automating the schema refactoring process, including the rewrite of SQL queries and updates~\cite{schema-survey,visser,prism1,prism2,chase1,chase2}. Among these works, the most related one is the {\sc Prism} project and its successor \textsc{Prism++}~\cite{prism1, prism2}. In addition to the original program and the source and target schemas, the {\sc Prism} approach requires the user to provide so-called \emph{Schema Modification Operators (SMOs)} that describe how  tables in the source schema are modified to tables in the target schema. The basic idea is to leverage these user-provided SMOs to rewrite  SQL queries using the well-known \emph{chase} and \emph{backchase} algorithms~\cite{chase1,chase2}. To deal with updates, they additionally require  the user to provide \emph{Integrity Constraint Modification Operators (ICMOs)} and ``translate'' updates  into queries. In contrast to the {\sc Prism} approach, our method does not require users to provide  modification operators expressed in a domain-specific language.
Although it is possible to explore the search space of SMOs and ICMOs to automate the generation of new database programs, we decided not to pursue this approach for two reasons: first, the rewriting technique in \textsc{Prism} requires these modification operators to be invertible, and, second, the search space of operator sequences is also potentially very large.

\paragraph{Analysis of database applications.} Over the past decade, there has been significant interest in analyzing, verifying, and testing database applications~\cite{empirical-study,wave1,wave2,verifas,zhendong07,smpsql,sparklite,mediator,sql-mc, rubicon,testing1, testing2, testing3, testing4}. For instance, the \textsc{WAVE}~\cite{wave1, wave2} and \textsc{VERIFAS}~\cite{verifas} projects aim to verify temporal properties of database applications, and recent work by
Itzhaky et al.~\cite{smpsql} proposes a technique to verify pre- and post-conditions of methods with embedded SQL statements. There has also been some work on model checking database applications~\cite{sql-mc, rubicon} as well as automatically generating test cases~\cite{testing1, testing2, testing3, testing4}.

Since our goal  is to synthesize a new version of the program that is \emph{equivalent} to a previous version, this paper is particularly related to verification techniques for checking equivalence~\cite{mediator,hottsql, cosette, sql-encoding, usemiring}. Most of these papers  focus  on equivalence between individual SQL queries~\cite{hottsql, cosette, sql-encoding, usemiring}. The only work that addresses the problem of verifying equivalence between entire database programs is \mediator, which automatically infers a bisimulation invariant between the two programs~\cite{mediator}.  As discussed in Section~\ref{sec:equiv}, we adopt the definition of equivalence proposed in that paper.
However, our synthesis technique does not simply add a CEGIS loop on top of \mediator's SMT theory because such an approach would require solving complex quantified  formulas over the theory of lists.  Instead, our approach performs enumerative search but uses minimum failing inputs to significantly prune the search space.

\paragraph{Program synthesis.} This paper is related to a long line of recent work on program synthesis~\cite{sketch1, sketch2, sketch3, sygus, synquid, flashfill, prose, neo, syngar, emina1, emina2, spr, deepcoder, mayur1, mayur2, mayur3, mcmc, relc, relc-concurrent, relish}.
While the goal of program synthesis is always to  produce a program that satisfies the given specification, different synthesizers use different forms of specifications, including  input-output examples~\cite{flashfill,prose,neo,deepcoder,syngar}, logical constraints~\cite{sketch1,sketch2,sketch3}, refinement types~\cite{synquid}, or a reference implementation~\cite{mcmc,relc,emina1}. Our technique belongs in the latter category in that it uses the original implementation as the specification.

Among existing techniques, our synthesis algorithm is particularly related to {\sc Neo}~\cite{neo}, which uses \emph{conflict-driven learning} to infer useful lemmas from failed synthesis attempts. Our sketch solving algorithm from Section~\ref{sec:sketch-complete} can also be viewed as performing some form of conflict driven learning in that it uses minimum failing inputs to rule out many programs that share the same root cause of failure as the currently explored one. However, our technique is much more lightweight compared to {\sc Neo} because it does not compute a minimum unsatisfiable core of the logical specification for the failing program. Instead, our technique exploits the observation that only a subset of the methods in a database program are necessary for proving disequivalence.

\paragraph{Synthesis for database programs.} In recent years, there have been several papers that apply program synthesis to SQL queries or database programs~\cite{scythe,morpheus,sqlizer,nalir,qbs,fiat}. For instance, {\sc Sqlizer}~\cite{sqlizer} synthesizes SQL queries from natural language, whereas {\sc Scythe}~\cite{scythe} and {\sc Morpheus}~\cite{morpheus} generate queries from examples. The {\sc QBS} system uses program synthesis to repair performance bugs in database applications~\cite{qbs}. Finally, {\sc Fiat}~\cite{fiat} performs deductive synthesis to generate SQL-like operations from declarative specifications. However, none of these techniques consider the problem of automatically migrating database programs in the presence of schema refactoring.

%% file: limit.tex
\section{Limitations}

In this section, we will explain and discuss some limitations of the \toolname tool.

First, \toolname cannot handle schema changes that are not expressible using our notion of value correspondence. For example, one can merge two columns ``first name'' and ``last name'' into a single column ``name'' and use string operations to extract first or last names in a query. These types of schema refactorings cannot be expressed using our definition of value correspondence. While it is relatively straightforward to expand our technique to a richer scope of value correspondences (e.g., by enriching the sketch language to include a set of predefined functions like {\tt concat}, {\tt split}, etc), this change would require a more sophisticated verifier that can reason about the semantics of built-in functions.

Second, \toolname does not synthesize database programs with control-flow constructs such as if statements and loops, because the underlying equivalence verifier~\cite{mediator} does not support database programs with those constructs. 

Third, the notion of equivalence considered in Section~\ref{sec:equiv} characterizes  \emph{behavioral equivalence} between database programs, which ensures that two corresponding sequences of transactions yield the same result. However, it does not enforce that the underlying data stored in the database is inserted or manipulated in particular ways. For example, \toolname may choose to delete from one or multiple tables when performing deletion as long as the new program satisfies the behavioral equivalence requirement. In some contexts, it may be desirable to adopt a stronger definition of equivalence than the one we consider in this paper.

%% file: concl.tex
\vspace{-10pt}
\section{Conclusion}
In this paper, we have studied the problem of automatically synthesizing database programs in the presence of schema refactoring. Our technique decomposes the synthesis procedure into three tasks, namely (i) lazy value correspondence enumeration, (ii) sketch generation from a candidate value correspondence, and (iii) sketch completion using conflict-driven learning with minimum failing inputs. We have implemented the proposed technique in a tool called \toolname and evaluated it on 20 schema refactoring benchmarks, including real-world scenarios taken from Github. Our evaluation shows that \toolname can automatically synthesize the new versions of all 20 benchmarks and indicates that the proposed technique would be useful to database application developers during the schema evolution process.

%% file: ack.tex
\vspace{-10pt}
\section*{Acknowledgments}
We would like to thank Shuvendu Lahiri for  initial discussions on the topic of this paper. We also thank Xinyu Wang, Kostas Ferles, Shankara Pailoor, and other UToPiA group members for their useful  feedback on earlier drafts. We would also like to thank our shepherd, Emina Torlak, and the anonymous reviewers for their constructive suggestions for improvement.
Finally, we are very grateful to the National Science Foundation for supporting this work under Grants \#1762299, \#1712067, and \#1811865.

%% file: appendixA.tex
\section{Theorems and Proofs}

\begin{theorem}[Soundness]
Given database program $\prog$ over source schema $\schema$ and a target schema $\schema'$, if procedure \textsc{Synthesize} ($\prog, \schema, \schema'$) returns $\prog'$ and $\prog' \neq \bot$, then $\prog' \simeq \prog$.
\end{theorem}
\begin{proof}
Observe that procedure \textsc{Synthesize} only returns non-$\bot$ program $\prog'$ if \textsc{CompleteSketch} returns $\prog'$ and $\prog' \neq \bot$.
Moreover, procedure \textsc{CompleteSketch} only returns non-$\bot$ program $\prog'$ if \textsf{Verify}($\prog, \prog'$) can verify the equivalence between $\prog$ and $\prog'$. By the soundness of oracle \textsf{Verify}, we have $\prog' \simeq \prog$.
\end{proof}

\begin{definition}[Completion of sketch]
Given a sketch $\sketch$ with holes $\vec{\hole}$, a program $\prog$ is called a \emph{completion} of $\sketch$, denoted by $\prog \in \gamma(\sketch)$, if and only if there exists a vector $\vec{e}$ where $e_i$ is in the domain of hole $\hole_i$ such that $\prog = \sketch[\vec{e} / \vec{\hole}]$.
\end{definition}

\begin{definition}[Structural isomorphism of programs]
Two database programs $\prog$ and $\prog'$ are \emph{structurally isomorphic} under value correspondence $\valueCorr$ if there is a mapping $M$ between the statements of $\prog$ and the statements of $\prog'$ such that
\begin{enumerate}
\item for each statement $S$ in $\prog$, $M(S) = S'_1; \ldots; S'_n$, where $S'_i$ is a statement in $\prog'$ and $S'_i$ is of the same type as $S$
\item for any pair of predicates $\phi$ in $S$ and $\phi'$ in $M(S)$, $\valueCorr \vdash \phi \sim \phi'$ as defined in Figure~\ref{fig:rules-pred-struct}
\end{enumerate}
\end{definition}

\begin{definition}[Database instance]
A database instance $\inst$ is a mapping from relation names to lists of tuples of corresponding relations in the database.
\end{definition}

We use the notation $\denot{U}_\inst$ to refer to the instance obtained by performing the update $U$ on $\inst$.
Similarly, $\denot{Q}_\inst$ refers to the result of query $Q$ in $\inst$.

\begin{definition}[Valid database instance]
A valid database instance of a program $\prog$ is one which can be obtained from the empty instance $\epsilon$ by performing a sequence of updates $U_1, \dots, U_n$ in $\prog$.
\end{definition}

\input{fig-rules-pred-struct}

\begin{definition}[Instance Mapping]
Given two structurally isomorphic programs $\prog$ and $\prog'$, the \emph{instance mapping} $\eta$ from valid instances in $\prog$ to valid instances in $\prog'$ is defined as follows
\begin{enumerate}
\item $\eta(\epsilon) = \epsilon$
\item for any update $U$ in $\prog$ and its corresponding update $U'$ in $\prog'$, $\eta(\denot{U}_\inst) = \denot{U'}_{\eta(\inst)}$
\end{enumerate}
\end{definition}

\begin{definition}[Strong equivalence]
Suppose we are given two structurally isomorphic programs $\prog, \prog'$ and the instance mapping $\eta$ between their valid instances $\inst, \inst'$. Programs $\prog$ and $\prog'$ are \emph{strongly equivalent} under value correspondence $\valueCorr$ if $\denot{Q}_\inst = \denot{Q'}_{\eta(\inst)}$ holds for any pair of queries $Q$ and $Q'$ related by $\valueCorr$.
\end{definition}

Intuitively, the notion of strong equivalence is different from the equivalence defined in Section~\ref{sec:equiv}, because the synthesis process only generates database programs that are structurally ``similar'' to the original program.

\begin{theorem}[Relative Completeness of Sketch Generation]
Suppose two structurally isomorphic programs $\prog$ and $\prog'$ are strongly equivalent under value correspondence $\valueCorr$. If \textsc{SketchGen} $(\prog, \valueCorr) = \sketch$, then $\prog' \in \gamma(\sketch)$.
\end{theorem}

\begin{proof}
We show that each statement must be in the completion of the sketch generated from the source program, so the final program must be as well. For each statement $S$ in $\prog$ and the corresponding statement $S'$ in $\prog'$, we perform structural induction on $S'$. Note that $S$ has to have the same type as each statement in $S'$.

\begin{itemize}[leftmargin=*]
\item
Base cases: $S'$ is a single statement. We perform structural induction on $S$.

    \begin{enumerate}
    \item
    Base case: $S$ is a join $\view$, $\valueCorr \vdash S \leadsto \sketch$. By the assumption of structural isomorphism, we know $S'$ must also be a join, say $\view'$.
    Since $\prog$ and $\prog'$ are related by value correspondence $\valueCorr$, we know that $\valueCorr \vdash \view \sim \view'$. Thus $\view'$ is exactly the generated sketch by the \textrm{Join} rule, and $\view' \in \gamma(\sketch)$.

    \item
    Inductive case: $S$ is of the form $\filter_\pred(Q)$, $\valueCorr \vdash S \leadsto \sketch$, then $S'$ is of the form $\filter_{\pred'}(Q')$.
    By inductive hypothesis, we know that $Q' \in \gamma(\sketch_0)$ where $\valueCorr \vdash Q \leadsto \sketch_0$.
    Since $\prog$ and $\prog'$ are structurally isomorphic, we have $\valueCorr \vdash \pred \sim \pred'$.
    By the structural isomorphism of predicates, replacing corresponding attributes $a$ in $\pred$ with $a' \in \valueCorr(a)$ yields a completion corresponding to $\pred'$, so $\filter_{\pred'}(Q') \in \gamma(\sketch)$.

    \item
    Inductive case: $S$ is of the form $\proj_{a_1, \dots, a_m}(Q)$, $\valueCorr \vdash S \leadsto \sketch$, then $S'$ is of the form $\proj_{a'_1, \ldots, a'_n}(Q')$.
    By inductive hypothesis, we know that $Q' \in \gamma(\sketch_0)$ where $\valueCorr \vdash Q \leadsto \sketch_0$.
    In addition, we have $m = n$, otherwise $\denot{Q}_{\inst} = \denot{Q'}_{\eta(\inst)}$ does not hold for any $\inst$, because $\denot{Q}_{\inst}$ and $\denot{Q'}_{\eta(\inst)}$ differ on the number of elements in the tuple.
    Given that $\prog$ and $\prog'$ are related by value correspondence $\valueCorr$, we have $a'_i \in \valueCorr(a_i)$ for $i \in [1, n]$ by the \textrm{Attr} rule. Thus, by the \textrm{Proj} rule, we know $\proj_{a'_1, \ldots, a'_n}(Q') \in \gamma(\sketch)$.

    \item
    Inductive case: $S$ is of the form $\ins(\view, \set{ a_1: v_1, \ldots, a_m: v_m})$, $\valueCorr \vdash S \leadsto \sketch$, then $S'$ is of the form $\ins(\view', \set{ a'_1: w_1, \ldots, a'_n: w_n})$.
    Let $r$ refer to the source tuple ($\set{ a_1: v_1, \ldots, a_m: v_m}$) in $S$ and $r'$ to the target tuple ($\set{ a'_1: w_1, \ldots, a'_n : w_n}$).
    Since $\prog$ and $\prog'$ are related by $\valueCorr$, we know that
    $\valueCorr \vdash \view \sim \view'$.
    For all valid instances $\inst$ after performing $S$ (i.e. there exists some $\inst_0$ such that $\inst = \denot{S}_{\inst_0}$), we know by the semantics of insertion that $r \in \denot{\proj_{a_1, \ldots, a_m}(\view)}_{\inst}$.
    Then since $\prog$ and $\prog'$ are strongly equivalent, if attribute $b_i \in \valueCorr(a_i)$ for $i \in [1, m]$ and $q$ is a tuple corresponding to $r$, $q \in \denot{\proj_{b_1, \ldots, b_m}(\view')}_{\eta(\inst)}$.
    Then $r'$ must be the tuple that is obtained by applying value correspondence $\valueCorr$ to tuple $r$. So $r' = \set{ b_1: v_1, \ldots, b_m: v_m}$ where $b_i \in \valueCorr(a_i)$, which is exactly a completion of the generated sketch $\sketch$ by the \textrm{Attr} and \textrm{Insert} rules.

    \item
    Inductive case: $S$ is of the form $\del(L, \view, \pred)$, $\valueCorr \vdash S \leadsto \sketch$, then $S'$ is of the form $\del(L', \view', \pred')$.
    Since $\prog$ and $\prog'$ are related by value correspondence $\valueCorr$, we know that $\valueCorr \vdash \view \sim \view'$.
    Let $R_1 = \denot{\filter_\pred(\view)}_{\inst}$ for some valid instance $\inst$, and let $R_2 = \denot{\filter_\pred(\view)}_{\inst'}$ where $\inst' = \denot{S}_{\inst}$.
    The semantics of delete require that the rows in $R_1$ but not in $R_2$ are removed.
    Let $R'_1$ and $R'_2$ be the results of the same queries in $\eta(\inst)$ and $\eta(\inst')$. The result of the queries must be equal, so we reduce it the filter case as discussed above.
    In addition, we have table list $L' \subseteq \view'$ because $S'$ is a well-formed statement. By the \textrm{Delete} rule, we know $\del(L', \view', \pred') \in \gamma(\sketch)$.

    \item
    Inductive case: $S$ is of the form $\upd(\view, \pred, a, v)$, $\valueCorr \vdash S \leadsto \sketch$, then $S'$ is of the form $\upd(\view', \pred', a', w)$.
    Since $\prog$ and $\prog'$ are related by value correspondence $\valueCorr$, we know that $\valueCorr \vdash \view \sim \view'$.
    Let $R_1 = \denot{\proj_a(\filter_\pred(\view))}_{\inst}$ for some valid instance $\inst$, and let $R_2 = \denot{\proj_a(\filter_\pred(\view))}_{\inst'}$ where $\inst' = \denot{S}_{\inst}$.
    The semantics of update require that the rows in $R_2$ are the same as those in $R_1$ but with the value set to $v$.
    Let $R'_1$ and $R'_2$ be the results of the same queries in $\eta(\inst)$ and $\eta(\inst')$. The result of the queries must be equal, so we know $a' \in \valueCorr(a)$ and $w = v$.
    By the \textrm{Update} rule and the case for filter discussed above, we know $\upd(\view', \pred', a', w) \in \gamma(\sketch)$.
    \end{enumerate}

\item
Inductive cases: $S'$ is of the form $S'_1; \ldots; S'_n$.
Let $S$ be the source statement corresponding to $S'$ and $\valueCorr \vdash S \leadsto \sketch$. We perform a case analysis on the type of statement $S$.
    \begin{enumerate}
    \item
    $S$ is a query statement.
    By the \textrm{Query} composition rule in Figure~\ref{fig:rules-compose}, we know $n = 1$. By inductive hypothesis, we have $S'_1 \in \gamma(\sketch)$. Thus, $S' \in \gamma(\sketch)$.

    \item
    $S$ is an update (\ins, \del, or \upd) statement.
    Based on the \textrm{Update} composition rules in Figure~\ref{fig:rules-compose} and the definition of $\bullet$ operator, we know that sketch $\sketch$ must be of the form $\sketch_1 \choice \ldots \choice (\sketch_{k1}; \ldots; \sketch_{kn}) \choice \ldots \choice \sketch_m$.
    By the inductive hypothesis, we have $S'_1 \in \gamma(\sketch_{k1}), \ldots, S'_n \in \gamma(\sketch_{kn})$. Thus, $S'_1; \ldots; S'_n \in \gamma(\sketch_{k1}; \ldots; \sketch_{kn})$.
    According to the semantics of $\choice$, we have $S'_1; \ldots; S'_n \in \gamma(\sketch)$.
    \end{enumerate}
\end{itemize}
\end{proof}

\begin{theorem}[Relative Completeness of Sketch Completion]\label{thm:complete}
Suppose the oracle \textsf{Verify} can verify equivalence between any pair of database programs and can provide minimum failing inputs if two programs are not equivalent. The \textsc{CompleteSketch} procedure can always find program $\prog' \in \gamma(\sketch)$ such that $\prog' \simeq \prog$ if such $\prog'$ exists.
\end{theorem}
\begin{proof}
We begin by showing that the \textsc{CompleteSketch} procedure is complete under the assumption that there exists a model $\model$ of the initial constraint $\Psi = \textsc{Encode}(\sketch)$ such that $\textsf{Instantiate}(\sketch, \model) = \prog'$.
Assume that $\prog'$ exists, and it is the first such $\prog'$. Let $\model$ be the model such that $\prog' = \textsf{Instantiate}(\sketch, \model)$.
First, observe that if $\model$ is not blocked in the last iteration, then the constraint $\Psi$ in that iteration must still be satisfiable, so the loop should not have terminated.
There must be some iteration in which $\model$ is blocked. But it only occurs after calling procedure $\textsf{Verify}(\prog, \prog')$, which should have returned \emph{true}. Thus, if such $\prog'$ exists, it will have been returned.

Next, we must show that constraint $\Psi$ itself is complete.
Let $\prog'$ be a completion of $\sketch$, i.e. $\prog' \in \gamma(\sketch)$.
For each hole $\hole_i$ that is assigned $e_j$ in program $\prog'$, assign variable $b_i^j$ to \emph{true} and all other variables $b_i^k$ ($k \ne j$) to \emph{false}.
Clearly, for every hole $\hole_i$, formula $\oplus(b_i^1, \dots, b_i^{i_n})$ evaluates to \emph{true}, as we have assigned a single expression to \emph{true}.
Then we have found an assignment that satisfies $\Psi$, and $\textsf{Instantiate}$ will produce the same completion.
\end{proof}

%% file: fig-rules-pred-struct.tex
\begin{figure}[!t]
\[
\begin{array}{c}
\irulelabel
{c \in \valueCorr(a) \quad d \in \valueCorr(b)}
{\valueCorr \vdash a \circ b \sim c \circ d}
{}
~
\irulelabel
{\valueCorr \vdash p \sim q}
{\valueCorr \vdash \neg p \sim \neg q}
{} \\ \ \\

\irulelabel
{\valueCorr \vdash p \sim r \quad \valueCorr \vdash q \sim s}
{\valueCorr \vdash p \land q \sim r \land s}
{}
~
\irulelabel
{\valueCorr \vdash p \sim r \quad \valueCorr \vdash q \sim s}
{\valueCorr \vdash p \lor q \sim r \lor s}
{} \\ \ \\

\end{array}
\]
\vspace{-15pt}
\caption{Inference rules for structural isomorphism of predicates under a value correspondence $\valueCorr$. The notation $\circ$ denotes a binary predicate (=, <, etc.)}
\label{fig:rules-pred-struct}
\vspace{-5pt}
\end{figure}